\documentclass{article} 
\usepackage{preprint, times}

\usepackage{amsmath,amsfonts,bm}

\def\eqref#1{equation~\ref{#1}}

\def\1{\bm{1}}

\DeclareMathAlphabet{\mathsfit}{\encodingdefault}{\sfdefault}{m}{sl}
\SetMathAlphabet{\mathsfit}{bold}{\encodingdefault}{\sfdefault}{bx}{n}

\newcommand{\R}{\mathbb{R}}

\newcommand{\softmax}{\mathrm{softmax}}

\DeclareMathOperator*{\argmax}{arg\,max}

\usepackage[hidelinks]{hyperref}
\usepackage{url}
\usepackage{amsmath,amssymb,amsthm,amsfonts}
\usepackage{tikz}
\usepackage{graphicx}
\usetikzlibrary{positioning,arrows.meta,shapes.geometric}

\newtheorem{theorem}{Theorem}[section]
\newtheorem{lemma}[theorem]{Lemma}
\newtheorem{proposition}[theorem]{Proposition}
\newtheorem{corollary}[theorem]{Corollary}

\newtheorem{remark}{Remark}[section]
\newtheorem{definition}{Definition}[section]
\newtheorem{example}{Example}[section]

\newcommand{\norm}[1]{\left\lVert #1 \right\rVert}

\def\last{\text{last}}
\def\out{\text{\rm{\rm{out}}}}

\def\F{\mathcal{F}}
\def\R{\mathbb{R}}

\def\Z{{\mathbb{Z}}}
\def\N{{\mathbb{N}}}

\def\TM{{\mathbb{T}}}

\def\Arpn{\mathbf{Arith}_{p,n}}
\def\Ar{\mathbf{Arith}}

\def\Alg{\mathbf{A}}

\DeclareMathOperator{\len}{len}
\DeclareMathOperator{\sen}{sen}
\DeclareMathOperator{\Poly}{Poly}

\DeclareMathOperator{\FNN}{FNN}
\DeclareMathOperator{\ATT}{ATT}
\DeclareMathOperator{\Relu}{ReLU}
\DeclareMathOperator{\CoT}{CoT}

\usepackage{booktabs}
\usepackage{pgfplots}
\usepackage{subcaption}
\pgfplotsset{compat=1.18}

\title{On the Turing Completeness of Transformers and Agents}
\author{
Yimu Qiao\textsuperscript{\rm 2, 3},
Lijia Yu\textsuperscript{\rm 1},
Ruichen Qiu\textsuperscript{\rm 2, 3}, 
Xiao-Shan Gao\textsuperscript{\rm 2, 3}\thanks{Corresponding author, xgao@mmrc.iss.ac.cn}\\
\textsuperscript{\rm 1}Institute of AI for Industries, Chinese Academy of Sciences\\
\textsuperscript{\rm 2}SKLMS, Academy of Mathematics and Systems Science, Chinese Academy of Sciences\\ 
 \textsuperscript{\rm 3}University of Chinese Academy of Sciences\\
}
\begin{document}
\maketitle
\begin{abstract}
Transformers have emerged as the dominant architecture in sequence modeling, achieving remarkable success in natural language processing and reasoning tasks. While existing literature has established the Turing completeness of transformers under bounded input length, the reasoning power of a single transformer operating on inputs of unbounded length is not fully explored.
In this paper, we theoretically investigate the reasoning limitations of a single transformer and the enhanced capabilities of agent systems. 
%
We show that a single fixed finite precision transformer cannot memorize  certain Turing machines with inputs of arbitrary length, such as the arithmetic; 
and a single fixed infinite precision transformer trained with a random algorithm is not Turing complete with probability one under reasonable conditions.
%
To overcome the limitation of a single transformer, we define a formal agent architecture consisting of decision, execution, and memory modules and show that for any Turing machine $\TM$, there exists an agent that can memorize $\TM$ and is computationally the same as $\TM$. Thus, agents are Turing complete.
\end{abstract}

\section{Introduction}

Transformers are the central architecture in modern sequence modeling and have achieved remarkable practical success in natural language processing, code generation, and mathematical reasoning. 
%
Recent theoretical studies have made substantial progress on the reasoning power of transformers. Log-precision transformers have been shown to be Turing complete \citep{Attention-TC,yx,jiang2025softmax} in a certain sense; constant size transformers have also been shown to be Turing complete \citep{li2026constant}. 
Moreover, log-precision  transformers with Chain-of-thought (CoT) are capable of solving any polynomial-time problem  \citep{merrill2023expressive}, and they are also able to handle the arithmetic task $\Arpn$ \citep{Feng-Cot2023}.
CoT has likewise been shown to improve the reasoning power of fixed-precision transformers \citep{ZhiyuanLi2024CoT,feng2024numerical}. On the negative side, it is proven that there exists a class of reasoning problems for which transformers, with and without CoT, have the same solving power~\citep{yu2025analyzing}.

As stated by \cite{cui2026position}, there is a limitation shared by much of the existing literature on transformer reasoning: the conclusions of most works are limited to finite length questions, such as \citep{Attention-TC,merrill2023expressive,Feng-Cot2023,yx,jiang2025softmax}. This means that for a given Turing machine $\TM$ and length limitation $n$, they can construct a transformer $\F_{\TM,n}$ depending on $n$ such that $\F_{\TM,n}$ can memorize $\TM$ on any input with a length not exceeding $n$. In other words,   $\{\F_{\TM,n}\}_{n=1}^{\infty}$ is Turing complete for simulating $\TM$, but not for a single transformer. However, people do not usually use a family of transformers to solve problems in practice.
\cite{mudarisov2025limitations} showed why transformers struggle to process long sequences, indicating that as the input length grows, the attention weights tend to become nearly uniform.

In order to understand the reasoning power of the transformer on inputs of arbitrary length, we have refined the relevant theories by answering the following question.

{\bf Question 1.
Is it possible for a single fixed transformer with CoT to memorize a Turing machine for inputs of arbitrary length?
}

For a Turing machine $\TM$, we say that a transformer $\F$ memorizes $\TM$ if, for any stoppable input $x$, it holds $\widehat{\F}_{\out}(x) = \TM(x)$. 
In this paper, we consider the case of unbounded length inputs for Turing machines, and we show that under reasonable conditions, a single transformer cannot achieve Turing completeness even with CoT. 
For finite precision situations, we employ an explicit constructive approach to demonstrate Turing incompleteness.
For infinite precision situations, we provide a tight sensitivity bound on how much a transformer’s output can change for two inputs that differ at only a single position (Theorems \ref{th-1} and \ref{th-2}), which is then used to obtain their Turing incompleteness. 
%
Overall, we obtain the following conclusions (Refer to Theorems \ref{zzdd1}, \ref{th-confd}, and \ref{ae}):
\begin{theorem}[Informal]
\label{th-m1}
(1) A single finite precision transformer is not Turing complete.
(2) A single infinite precision transformer is not Turing complete under positive confidence.
(3) A single infinite precision transformer trained with a random algorithm is not Turing complete with probability one under reasonable conditions.
\end{theorem}






Theorem \ref{th-m1} shows that single transformer with CoT cannot memorize certain Turing machines on unbounded length inputs. 
On the other hand, agents based on transformer are widely used to solve complex tasks, such as mathematical reasoning. Through the combination of LLMs, skills and tools,  the capabilities of agents can be significantly improved. 
In light of the negative results in Theorem \ref{th-m1}, it is interesting to examine the reasoning capabilities of LLM-based agents.

{\bf Question 2. 
Are LLM-based agents Turing complete?
}

It was shown that memory augmented LLMs are computationally universal \citep{Schuurmans2023}.
Recently, memory augmented LLMs have been shown to be Turing complete with explicit constructions \citep{chen2026neural}.
In this paper, we focus on agents that are more realistic and show their Turing completeness. 
Specifically, we arrive at the following conclusions (Refer to Theorem \ref{agent}):

\begin{theorem}[Informal]
\label{th-m2}
We define an agent architecture consisting of finite precision transformers as decision and execution modules, a memory module, and linear-time computable tools, such that for any Turing machine $\TM$, there exists an agent that can memorize $\TM$ and is computationally the same as $\TM$.
\end{theorem}

Theorem \ref{th-m2} demonstrates that combining simple tool calls with a loop execution pattern is significantly more effective than general CoT for solving long problems. 

Finally, we experimentally verify our theoretical conclusions using several widely used LLMs. We should emphasize that this is primarily a theory paper, and the experiments serve as proof-of-concept demonstrations rather than comprehensive large-scale benchmarks.

In summary, the main contributions of this paper are threefold. 
%
\begin{enumerate}
\item 
We show that a fixed finite precision transformer cannot memorize certain Turing machines, such as arithmetic, and thus is not Turing complete.  
Infinite precision transformers are also not Turing complete with positive confidence or with  probability one under certain conditions.  
\item
We show that an agent built on finite-precision transformers can memorize Turing machines with inputs of unbounded length and thus is Turing complete. 
\item
We validate our theoretical results experimentally on some common LLMs and agents.
\end{enumerate}

\section{Related Work}

\noindent{\bf Expressive power and Turing completeness.}
A substantial line of work studies the computational expressiveness of transformers through Turing-machine simulation and circuit complexity. Early results established Turing completeness for encoder–decoder transformers and analyzed the roles of positional encodings, masking, and architectural components \citep{Attention-TC,bhattamishra2020computational,yx}. Subsequent work extended these results to decoder-only architectures and transformers with CoT reasoning \citep{roberts2023powerful,merrill2023expressive}. More recent constructions establish Turing completeness for constant-bit-size transformers and softmax CoT transformers under suitable assumptions \citep{li2026constant,li2025efficientturingmachinesimulation,jiang2025softmax}. Circuit-complexity analyzes further characterize how transformer expressiveness depends on numerical precision, depth, and the number of intermediate reasoning steps \citep{merrill2022saturated,merrill2023parallelism,chiang2025uniform,merrill2025logdepth,ZhiyuanLi2024CoT}. Related work also studies limitations induced by attention normalization, including the selective capacity and token separation of softmax attention \citep{mudarisov2025limitations}.
%
A line of related research studies whether  the ability learned on bounded-length inputs extends to larger lengths 
\citep{kazemnejad2023impactpositionalencodinglength,ICLR2024_45ed1a72,wang-etal-2024-length,length-generalization-transformers,chen2026position,yang2026length}. 

\noindent{\bf Agents and external computation.}
LLM-based agents combine repeated model invocation with external state, context management, and auxiliary tool usage \citep{yao2023iclr-react,NEURIPS2023_d842425e}. External memory and context-management mechanisms have been used as additional computational resources for transformer systems \citep{chen2026neural}. The theoretical work further analyzes how the computational power of such systems depends on the external memory and context-management mechanisms coupled with a fixed language model \citep{Schuurmans2023,cui2026position}.
A detailed comparison of our work with \citep{chen2026neural} is given in Remark \ref{r34}.

\section{Notation, Model, and Basic Assumptions}
\label{nmb}
In this section, we provide the basic notations. 
More details can be found in Appendix \ref{app-N}.


\subsection{Autoregressive Transformer}

Let $\Sigma=\{\sigma_i\}_{i=1}^T$ be the vocabulary or alphabet set, and let $\Sigma^*$ denote the set of all finite sequences over $\Sigma$. The symbol $\sigma_0$ is used only as an end-of-output token and does not appear in intermediate positions of an input sequence. For any sequence $x\in\Sigma^*$, let $\len(x)$ denote its length, let $x[i]$ denote its $i$-th symbol, and let $\last(x)=x[\len(x)]$ denote its last symbol. For any two sequences $x_1,x_2\in\Sigma^*$, write $x_1\oplus x_2$ for their concatenation. If $S\subset\Sigma^*$ is finite, write $\len(S)=\max_{x\in S}\len(x)$.


This article considers a transformer $\F=E_{\rm{out}}\circ \F_L\circ\F_{L-1}\circ\cdots\circ\F_1\circ E_{\rm{emb}}:\Sigma^*\to\R^{T+1}$, where $E_{\rm{emb}}$ is the embedding layer, $E_{\rm{out}}$ is the output layer, $\F_l$ is the hidden layer with the RoPE position embedding $R$, position mask $M$, and multi-head attention, as follows:
\vspace{-2mm}
\begin{align*}
&\F_{\ell,1}(x)=\sum_{h=1}^H\softmax(R(xQ_{\ell, h}K_{\ell, h}x^\top)+M)\,xV_{\ell, h},\\
&\F_{\ell}(x)=x+\F_{\ell,1}(x)+\FNN(x+\F_{\ell,1}(x)).   
\end{align*}
This  is a commonly used transformer architecture. 
We focus on token-by-token autoregressive generation as the output mode, and we denote the transformer’s output for an input $x$ by $\widehat\F(x)$. 

\subsection{Memorization of Turing machine by Transformer}

In this section, we define how a transformer could memorize a Turing machine. 

We represent a Turing machine as $\mathbb{T} = (Q,\Sigma_1,\Sigma_2,\delta)$, where $Q$ is the state set, $\Sigma_1$ is the set of symbols that can be input, $\Sigma_2$ is the set of symbols that can appear in the calculation process and output, and $\delta$ is the transfer function.  The empty-symbol $B$ is in $\Sigma_2\setminus\Sigma_1$. 
For a Turing machine $\TM$, we call an input $x \in \Sigma_1^*$ {\em stoppable} for $\TM$ if $\TM(x)$ halts after a finite number of steps. 
For any stoppable input $x$, we denote $t_\TM(x)$ to be the number of steps that $\TM$ performs on input $x$ before halting, $t_\TM(x)$ is also referred to as the {\em running time} of the Turing machine $\TM$ for solving $x$. 

To guarantee that a transformer can reproduce the entire output of a Turing machine, we require its alphabet $\Sigma$ to contain a special symbol, denoted by $\sigma_T$. The symbol $\sigma_T$ serves as a delimiter that separates the CoT from the final output. For any generated sequence $\widehat{\F}(x)$ containing $\sigma_T$, we use the last occurrence of $\sigma_T$ as the delimiter and write $\widehat{\F}(x)=\CoT(\F(x))\oplus\sigma_T\oplus\widehat{\F}_{\out}(x)$, where $\CoT(\F(x))$ is the prefix preceding the last occurrence of $\sigma_T$, and $\widehat{\F}_{\out}(x)$ is the suffix following it. Denote $\Sigma_0=\Sigma\setminus\{\sigma_T\}$.

\begin{definition}
We say that a transformer $\F$ memorizes a Turing machine $\TM=(Q,\Sigma_1,\Sigma_2,\delta)$ with a CoT if $\Sigma_2\subset\Sigma_0$ and for every stoppable input $x\in\Sigma_1^*$ for $\TM$, $\widehat{\F}(x)$ contains $\sigma_T$ and $\widehat{\F}_{\out}(x)=\TM(x)$. The CoT length is defined as $C(x)=|\CoT(\F(x))|=|\widehat{\F}(x)|-|\widehat{\F}_{\out}(x)|-1$.
\end{definition}





For simplicity, when we later say the transformer memorizes a Turing machine, we mean it does so using CoT and $\Sigma_2\subset \Sigma_0$.
Let us use an example to illustrate these notions. 
\begin{example}\label{eg:mem_turing}
Let 
$\Sigma_1=\{0,1,2,3,4,5,6,7,8,9,+,-,*,/,(,),|\}$, 
$\Sigma_2=\Sigma_1\cup\{=,B\}$, and
$\TM_{\Ar}=(Q,\Sigma_1,\Sigma_2,\delta)$ be the Turing machine for solving the arithmetic of rational numbers, with $n|m$ to represent a rational number $\frac{n}{m}$. 
For instance, if $x=2+2*(5-1)+3$, then $\TM_{\Ar}(x)=13$. 
Let $\sigma_T$ be `$=$'. For a transformer $\F$, if $\widehat{\F}(x)$ is `$=2+2*4+3=2+8+3=10+3=13$' with final output $\sigma_0$, then $\widehat{\F}_{\out}(x)$ is $13$, which is the sequence  between the last $=$ and $\sigma_0$. 
%
\end{example}

%


\section{Limitations of Transformer in Memorizing Turing machine}
\label{sec-Tans}

Prior studies have demonstrated that a CoT transformer is capable of memorizing a Turing machine for inputs with \emph{bounded length}~\citep{Attention-TC,merrill2023expressive,yu2025analyzing,yx}, as detailed in the Appendix \ref{old}. 
In this section, we show that a single transformer cannot memorize most Turing machines for input with \emph{unbounded length}. 


\subsection{A Finite-precision Transformer is not Turing Complete}
\label{ffa}
In this section, we demonstrate that a single fixed transformer with finite precision is unable to solve simple algorithmic tasks such as $\Ar$, which admit inputs of unbounded length. Therefore, such a model is not Turing complete.  Firstly, we define the transformer with finite precision.
\begin{definition}
By saying that a transformer $\F$ has $q$-precision, we mean that all its parameters, as well as the matrix calculation, Softmax calculation, and the position embedding calculation in the hidden-layer output, are represented with $q$ decimal real numbers. We denote a transformer with $q$-precision by $\F_q$. Unless otherwise specified, it is assumed that $\F$ has un-limited precision in this paper.
More details about the finite precision transformer can be found in Appendix \ref{precision}. 
\end{definition}

Given any sequence $x$ and $n\in\N_+$, we define $x^n$ to be the concatenation of $n$ copies of $x$.
Below,  we define a class of Turing machines that cannot be memorized by fixed-precision transformers.
\begin{definition}
A Turing machine $\TM=(Q,\Sigma_1,\Sigma_2,\delta)$ is called {\em non-converging} if there exist $x,y,z\in\Sigma_1^*$ such that every $x\oplus y^n\oplus z$, $n\in\mathbb N_+$, is a stoppable input for $T$, and $\TM(x\oplus y^n\oplus z)\ne\TM(x\oplus y^m\oplus z)$ for any $m,n\in\Z_+,m\ne n$.
\end{definition}

Based on this definition, we establish the following result about the Turing incompleteness for finite precision transformers:

\begin{theorem}
\label{zzdd1}
Let $\TM$ be a non-converging Turing machine. Then, for any $q\in\N_+$ and any fixed precision transformer $\F_q$, $\F_q$ cannot memorize $\TM$. Consequently, a fixed-precision transformer cannot be Turing complete.
\end{theorem}
\begin{proof}
    In processing long inputs, a Turing machine is limited only by memory capacity. In contrast, Transformers must also contend with precision limits. This is particularly evident in position embeddings, where precision loss in trigonometric calculations with sequence length leads to cyclic phenomena in RoPE encoding. Additionally, the attention mechanism is susceptible to numerical instability; the exponential operations can produce values so minute that they underflow to zero, effectively truncating the computation.
    The complete proof is in Appendix~\ref{app:proof_thm3}.
\end{proof}

In practice, a large portion of practical Turing machines is non-converging. For instance, $\TM_\Ar$ is obviously non-converging, though any finite precision transformer $\F_q$ cannot memorize $\TM_{\Ar}$.



\subsection{Infinite-precision Transformer is not Turing Complete under some conditions}
\label{adfdfgsef}
In this section, we discuss infinite precision transformers, which are more powerful than finite precision transformers. 
We first present a key fact in Section \ref{sss1} and use that fact to obtain the main result about Turing incompleteness in Sections \ref{in-pr} and \ref{mtmfae}. 
The application scopes of the theorems in this paper are clarified in Appendix \ref{asct}.


\subsubsection{A tight bound of value propagation for two distinct tokens}
\label{sss1}
In this section, we derive a tight bound on the transformer outputs for the embeddings of two tokens, which will be used to prove the Turing incompleteness of  infinite-precision transformers in the subsequent section.
We first define a local sensitivity measure for transformers, which quantifies the effect of modifying a single input token at a non-output position on the final representation.

\begin{definition}
\label{def1}
Let $\mathbb{B}(x,k)$ be the set of  $x'$ such that $\len(x)=\len(x') = n$ and $x[j]\ne x'[j]$ for no more than $k$ values of $j \in [n-1]$, and $x[n] = x'[n]$. That is, \(x'\) differs from \(x\) by no more than \(k\) non-output positions. 
For a sequence  $x$ and a transformer $\F$, the {\em sensitivity} of $\F$ with respect to $x$ is defined as
$$\sen(\F,x)=\max_{x'\in \mathbb{B}(x,1)}||\F(x)-\F(x')||_2.$$ 
\end{definition}

Now we present the following key technical result.
\begin{theorem}
\label{th-1}
For any input $x$ and any transformer $\F$ previously defined with depth $L$, we have $$\sen(\F,x)\le C_\F\frac{\ln^{L-1}(\len(x))}{\len(x)},$$ 
where $C_\F$ is a constant that depends only on $\F$.
\end{theorem}
\begin{proof}
    Consider two inputs differing only at position \(i\). Since the hidden states and attention logits of a fixed transformer are uniformly bounded, each causal attention weight at position \(j\) is at most \(O(1/j)\). 
    Hence, a local perturbation is diluted to \(O(1/j)\) after the first layer. Inductively, if the perturbation at layer \(\ell-1\) is \(O((\ln j)^{\ell-2}/j)\), then the attention aggregation contributes one additional logarithmic factor.
    Using the Lipschitz continuity of the remaining components and taking \(j=\operatorname{len}(x)\) yields $\operatorname{sen}(F,x)\le C_F\frac{(\ln \operatorname{len}(x))^{L-1}}{\operatorname{len}(x)}.$
    The complete proof is in Appendix~\ref{app:proof_thm4}.
\end{proof}

By the triangle inequality, Theorem \ref{th-1} can be extended to $\mathbb{B}(x,K)$.

\begin{corollary}
\label{cor:k-perturb}
There exists a constant \(C_{\F}>0\), depending only on \(\F\), such that for every input \(x\), every \(K\in\N_+\), and every \(x'\in \mathbb{B}(x,K)\),
\[
\|\F(x)-\F(x')\|
\le
K C_{\F}\frac{\ln^{L-1}(\len(x))}{\len(x)}.
\]
\end{corollary}

The bound in Theorem \ref{th-1} is tight up to constant factors, as shown below.
\begin{theorem}
\label{th-2}
For any $L$, there exists a transformer $\F$ of depth $L$ and a constant $c_\F>0$ such that, for any $M\in\N_{+}$, there is an input $x$ with $\len(x)\ge M$ satisfying
\[
\sen(\F,x)\ge c_\F\frac{\ln^{L-1}(\len(x))}{\len(x)}.
\]
\end{theorem}
\vspace{-2mm}
\begin{proof}
We construct a fixed transformer with uniform causal attention so that each layer averages the perturbation over all preceding positions. Starting from a perturbation at the first position, this repeated averaging accumulates one logarithmic factor per layer, and after \(L\) layers, the perturbation at the final position matches the upper bound in Theorem ~\ref{th-1} up to a constant factor. 
Refer to Appendix~\ref{app:proof_thm5} for the whole proof.
\end{proof}

\subsubsection{Infinite-precision transformer is not Turing Complete under  Confidence}
\label{in-pr}

In this section, we show that if the output of the transformer is required to be under a certain threshold, then even  infinite-precision transformers are not Turing complete.
Firstly, we define that:
\begin{definition}
$\F(x)$ is said to  output $y$ with {\em confidence} $\epsilon>0$, if $y=\widehat{\F}(x)$,
$$\F_{m_y}(x\oplus y[1:j-1])-\max_{k\ne {m_y}}F_k(x\oplus y[1:j-1])>\epsilon$$ 
for any $j\in[\len(y)]$, where $m_y=\argmax _{k\in[T+1]}\F_k(x\oplus y[1:j-1])$.    

We say that a transformer $\F$ can memorize a Turing machine $\TM$ with confidence $c$ if $\widehat{\F}_{\out}(x) = \TM(x)$ with confidence $c$  for any stoppable input $x$ of $\TM$.
\end{definition}


We define a class of Turing machines that cannot be memorized by unlimited precision transformers with positive confidence.
\begin{definition}
An input $x\in\Sigma_1^*$ is called $(C,\alpha)$-sensitive if $x\in\Sigma_1^*$ is  stoppable and there is an $x_1\in\Sigma_1^*\cap \mathbb{B}(x,\lceil C\len(x)^\alpha \rceil)$ that is a stoppable input such that $\TM(x)\ne \TM(x_1)$.

A Turing machine $\TM=(Q,\Sigma_1,\Sigma_2,\delta)$ is said to be \emph{non-tame} if there exist constants $C>0$ and $0\le\alpha<1$ such that $\TM$ has infinitely many $(C,\alpha)$-sensitive inputs.
\end{definition}

Based on this kind of Turing machine, we establish the following result about Turing incompleteness for
infinite precision transformers:
\begin{theorem}
\label{th-confd}
A non-tame Turing machine $\TM$ cannot be memorized by any fixed transformer $\F$ with any given confidence $c>0$.
\end{theorem}
\begin{proof}
Corollary \ref{cor:k-perturb} implies that, for any $\epsilon>0$, there exists $N$ such that for all $x$ for which $\len(x)>N$, $\sen(\F,x)<\epsilon$. 
In other words, once the sequence is sufficiently long, changing tokens at some non-output positions has only an infinitesimal influence on the final output position.
The complete proof is in the Appendix~\ref{app:proof_thm6}
\end{proof}

Most practically relevant Turing machines are non-tame and, therefore, by Theorem \ref{th-confd}, cannot be memorized by a single transformer with any fixed confidence.
In particular, $\TM_\Ar$ is non-tame since it contains infinitely many $(1,0)$-sensitive inputs.

\begin{remark}
The incompleteness with confidence result carries practical implications.
During the actual inference, the outputs of LLMs are probabilistic. 
To control generation diversity, sampling strategies such as Top-$K$ or Top-$p$ are typically employed.
When the input is very long, the transformer produces low-confidence results, which can prevent it from generating a correct final answer with a high probability.
This observation will be theoretically validated in the subsequent section.
\end{remark}

\subsubsection{Memorization of Turing machine  Fails Almost Everywhere}
\label{mtmfae}

The confidence-based result is sensible for a transformer using probabilistic sampling, but it does not hold for a transformer using greedy sampling. Therefore, in this section we drop the confidence-margin assumption. The trade-off is that the conclusion becomes measure-theoretic: the probability that a trained transformer is Turing complete is zero.


First, we define a distribution satisfied by the transformer.
\begin{definition}
For the fixed vocabulary set $\Sigma$, dataset $S\subset \Sigma^*\times\Sigma^*$, depth $D$, width $W$, head $H$,   and a random algorithm $\Alg$, let $\Alg(S,D,W,H,\Sigma)$ be the distribution of parameters of transformers with depth $D$, width $W$, head $H$, and vocabulary $\Sigma$ obtained after training on $S$ with $\Alg$. 
\end{definition}

Then, for a Turing machine, we define the following space:

\begin{definition}
\label{flly}
%
For a $G:\Sigma^*\to2^{\Sigma^*}$  and an $f:\N\to\R$, we say that \(G\) is an \(f\)-number of rule if, for every \(x\), \( |G(x)| \le f(\len(x))\).

Let $G$ be a $f$-number of rule and $\TM$ be a Turing machine. 
We say a transformer $\F\in M_{G}(\TM)$ if $\F$ can memorize  $\TM$ and $\CoT({\F}(x))\in G(x)$ for any $x$. In other words, there exist at most $f(\len(x))$ ways to generate the CoT for ${\F}(x)$ when it memorizes $\TM$.
%
\end{definition}

Then, we have the following result:
\begin{theorem}
\label{ae}
  For arbitrary $S,D,W,H,\Sigma$, assume that the following conditions are satisfied.
    
(1) $\TM=(Q,\Sigma_1,\Sigma_2,\delta)$ is a $(C,\alpha)$ non-tame Turing machine and $\Sigma_2\subset \Sigma_0$.

(2) $\Alg(S,D,W,H,\Sigma)$ is an absolutely continuous distribution.

(3) $G$ is a $f$-number of rules, and $f$ satisfies $\lim_{n\to\infty}\frac{f(n)(\ln n)^{(D-1)/2}}{\sqrt{n^{1-\alpha}}}=0$.

Then, $\Pr_{\F\sim \Alg(S,D,W,H,\Sigma)}(\F\in M_G(\TM))=0$.
That is, if fewer than $\widetilde{o}(\sqrt{n^{1-\alpha}})$ CoTs are allowed for inputs of length $n$, then the probability for $\F$ trained by algorithm $\Alg$ to memorize $\TM$ is zero.
\end{theorem}
\begin{proof}
    Theorem ~\ref{th-1} bounds the change of the logits between two nearby inputs, but this alone does not determine whether their predicted tokens are different. We therefore fix all parameters except the output bias so that the token-selection condition can be expressed directly as a constraint on the bias difference. For each pair of admissible generation sequences, consider their first disagreement. Theorem ~\ref{th-1} then restricts the corresponding bias difference to a thin interval, giving a thin slab in the bias space. Since there are at most $f(n)^2$ such pairs, the feasible parameter set is covered by $f(n)^2$ slabs, whose total measure tends to zero under the assumed growth condition. Hence, the memorization parameter set has probability zero.
    Refer to Appendix~\ref{app:proof_thm7} for the whole proof.
\end{proof}

The three conditions of Theorem \ref{ae} are reasonable.
The first two conditions are clearly reasonable.
For condition (3), the number of CoT could be considered the number of ways to solve the problem, which need not be large, such as the constructed CoT in \citep{yx}.

For $\TM_\Ar$, we have the following stronger result (the proof is given in Appendix \ref{app-Tarithae}).
\begin{corollary}
\label{Tarithae}
For $\TM_\Ar$, we restrict the generation rules $G$ to a fixed deterministic one-step calculation rule. Then, if $\Alg(S,D,W,H,\Sigma)$ is an absolutely continuous distribution, then $\Pr_{\F\sim \Alg(S,D,W,H,\Sigma)}(\F\in M_G(\TM_\Ar))=0$.
\end{corollary}

\vspace{-2mm}
\section{Memorizing Turing machine by Agent}
\label{sec-inf}
\vspace{-2mm}
In this section, we formally define a class of agents that incorporates multiple transformers along with simple tools and show that such agents are Turing complete.

\subsection{Definition of Agent}
\label{sec-agent}

First, we introduce the agent $A$ and how it operates.
\begin{definition}
\label{def:agent_frame}
An agent $A=(\F_d,\F_c,\mathbf{M},v_{\rm{end}},v_d,v_r,v_{\out})$ defined over the vocabulary set $\Sigma$ is specified as follows:

 (a1) $\widehat{\F}_d: \Sigma^*\to\Sigma^*$ is a transformer that serves as the decision module of the agent $A$;

(a2) $\widehat{\F}_c: \Sigma^*\to\Sigma^*$ is a transformer that serves as the execution module of the agent $A$;

(a3) $\mathbf{M}$ is the memory module of the agent $A$, which is a space to store sequences in $\Sigma^*$;

(a4) $v_{\rm{end}}:    \Sigma^*\to\{0,1\}$, $v_{\out}: \Sigma^*\to\Sigma^*$, $v_d:\Sigma^*\times\Sigma^*\to\Sigma^*$, and $v_r:\Sigma^*\times\Sigma^*\to\Sigma^*$ represent the tools used by the agent for stop checking, output, module connection, and memory writing.
 
   Given an input $x\in\Sigma^*$ to $A$, such an agent operates in the following manner:

(b1) Let $M$ be the sequence in the $\mathbf{M}$, and let $M=x$ at first; 

(b2) Calculate $v_1=\widehat{\F}_d(M)$. If $v_{\rm{end}}(v_1)=0$, go to (b3); otherwise, go to (b5);

(b3) Calculate $v_2=\widehat{\F}_{c}(v_d(M,v_1))$;

(b4) $M=v_r(M,v_2)$ and go to (b2);

(b5) Output $A(x)=v_{\out}(M)$.
\end{definition}

This agent is defined such that its input space, memory space, and output space all lie within $\Sigma^*$. Both the decision and execution models are defined by the corresponding transformer, and $v_d, v_{\rm{end}}, v_{\out}$ and $v_r$ represent external tools that the agent can call. The five steps (b1)-(b5) represent the five common parts in the operation of an agent: initialization, decision-making module operation, stop determination, execution module operation, writing in the memory module, and final output.
For $q\in\Z_+\cup\{\infty\}$, $A_q$ means that each transformer in the agent has a precision of $q$. We further define that:

\begin{definition}
We say that an agent $A_q$ can memorize a Turing machine $\TM=(Q,\Sigma_1,\Sigma_2,\delta)$ if it holds that $\Sigma_2\subset\Sigma_0$ and $A_q(x) = \TM(x)$ for every stoppable input $x$ of $\TM$.
\end{definition}

\subsection{Agent is Turing Complete.}
\label{atc}

We show that the agent defined in Section \ref{sec-agent} is Turing complete.
\begin{theorem}
\label{agent}
Assume that $|\Sigma|\ge|\Sigma_2|+6$, there are functions $v_d$, $v_{\out}$, $v_{\rm{end}}$, and $v_r$ computable in linear time, such that for any Turing machine $\TM=(Q,\Sigma_1,\Sigma_2,\delta)$, one can construct an agent $A_q$ which is able to memorize $\TM$. Moreover, $A_q$ satisfies the following:

(1) The transformers $\F_d,\F_c$ have constant depth $O(1)$, width $O(|Q|\cdot|\Sigma_2|)$, number of heads $O(\ln|Q|)$, and numerical precision $q=O(\ln|Q|)$.

(2) The agent executes step (b2) at most $t_\TM(x)+1$ times; the memory module requires at most $O(t_{\TM}(x)+\len(x)+\log_2
(|Q|))$ space; the length of output for $\F_d,\F_c$ in each reasoning is $O(\ln|Q|)$.

As a consequence, the agent $A_q$ is computationally the same as $\TM$.
\end{theorem}
\begin{proof}
Each cycle of the agent memorizes one step of a Turing machine. 
The decision module decides whether to halt by checking for the existence of a state $q_s$.
The execution module repeatedly memorizes the Turing machine’s transfer function and saves the outcomes in the memory module.
Finally, the output is produced according to the content in the memory module. 
We show the complete proof in Appendix~\ref{app:proof_thm8}.
\end{proof}

In terms of time and space complexity, for any input $x$, such a agent needs $O(\Poly(|Q|\cdot|\Sigma_2|)(t_\TM(x)+\len(x)))$-time and $O(t_\TM(x)+\len(x)+\ln|Q|)$ space to memory $\TM(x)$. This implies that the agent can memorize $\TM$ with polynomial overhead about $|Q|,|\Sigma_2|$, but not $t_\TM(x)$. Therefore, such an agent is computationally the same as the Turing machine if $|Q|$ and $|\Sigma_2|$ are considered constants. From this, it can be seen that by using simple but appropriate tools that don't depend on the target Turing machine, agents can achieve Turing completeness. 

Moreover, the two transformers $\F_c,\F_d$ in the agent can be the same one. 
We can consolidate   $\F_c$ and $\F_d$ into a single transformer, $\F_s$, and include a suitable prompt so that $\F_s$ can distinguish whether it is performing decision-making or execution.
The detail is given in the Appendix \ref{thcap}.





In summary, this section states that for agents, by choosing appropriate tools (that do not depend on the specific target Turing machine), we can construct the decision model and the execution model to memorize any Turing machine.

\begin{remark}
\label{r34}
In comparison with \citep{chen2026neural}, our results exhibit the following differences.
Our agent maps directly to real-world agent architectures: the decision transformer $\F_d$  corresponds to the ``Planner" that orchestrates workflows; the execution transformer $\F_c$  mirrors the ``Worker" models; and the memory $M$ represents the context management mechanism. 
We use finite-precision transformers, and the agent is computationally the same as $\TM$.
\end{remark}
\section{Experiment}
\label{sec:experiment}
Our theoretical analysis characterizes a computational separation between a single transformer and an agent that can repeatedly invoke transformers while maintaining a memory module. 
To complement our theoretical analysis, we empirically examine whether iterative LLM execution with external memory and tools improves reliability in increasingly long computations.  
We compare one-shot inference with an agent that performs the computation step by step, focusing on how their accuracy changes as the number of required operations grows.

\subsection{Experiment Settings}
We use arithmetic $\Ar$ because it provides unambiguous answers and direct control over computation length. 
The benchmark contains 1,300 expressions spanning 1 to 30 arithmetic operators and covers addition, subtraction, multiplication, division, and mixed operations. 
For each evaluated language model, we compare two inference architectures.
Refer to Appendix~\ref{app:experimental-details} for further details.


\textbf{Single model.}
The model receives the complete expression in one request and must return its final answer before that request terminates. It may produce textual reasoning within the response, but it cannot call a tool, update external state, or invoke the model again.

\textbf{Agent.}
Each component of the implementation directly instantiates one element of our theoretical agent framework in Definition~\ref{def:agent_frame}. The memory module $\mathbf{M}_t$ stores the expression in iteration $t$, where $\mathbf{M}_0$ keeps the original input expression. 
The decision component $\F_d$ is an LLM that selects the identifier of the next operator.
The linear-time read function $v_d$ retrieves that operator and its two adjacent operands from $\mathbf{M}_t$.
The execution module $\F_c$ is an LLM that computes this local binary operation.  
Finally, the linear-time replacement function $v_r$ writes the value provided by $\F_c$ back to memory to construct $\mathbf{M}_{t+1}$.  
The process repeats until one value remains.  
Neither $v_d$ nor $v_r$ performs arithmetic, chooses an operator, or provides correctness feedback.
We show a real example of the agent workflow in Appendix~\ref{app:agent-example}.

\begin{table}[t]
    \footnotesize
    \centering
    \begin{minipage}[b]{.49\textwidth}
    \footnotesize
    \centering
    \begin{tikzpicture}
        \begin{axis}[
            width=7.2cm,
            height=4.2cm,
            axis lines=left,
            xlabel={\# Arithmetic Operators},
            ylabel={Accuracy},
            tick align=inside,
            ymajorgrids=true,
            grid style={dashed, gray!30},
            xmin=0, xmax=32,
            ymin=0, ymax=1,
            legend style={
                draw=none,
                at={(.9,1.25)},
                font=\scriptsize, 
                legend columns=3,
                column sep=2pt,
            },
            legend image post style={xscale=.5}
        ];
    
        \addplot[blue, dashed, thick, no marks, forget plot] coordinates {
            (1,0.91)
            (2,0.6)
            (3,0.35)
            (4,0.2)
            (5,0.15)
            (6,0.07)
            (8,0.14)
            (10,0.07)
            (12,0.08)
            (15,0.14)
            (20,0.03)
            (25,0.02)
            (30,0)
        };
        \addplot[blue, thick, no marks] coordinates {
            (1,0.86)
            (2,0.75)
            (3,0.63)
            (4,0.54)
            (5,0.49)
            (6,0.52)
            (8,0.27)
            (10,0.22)
            (12,0.23)
            (15,0.21)
            (20,0.2)
            (25,0.21)
            (30,0.23)
        };
        \addplot[red, dashed, thick, no marks, forget plot] coordinates {
            (1,1)
            (2,0.98)
            (3,0.97)
            (4,0.94)
            (5,0.89)
            (6,0.8)
            (8,0.77)
            (10,0.46)
            (12,0.35)
            (15,0.28)
            (20,0.28)
            (25,0.2)
            (30,0.24)
        };
        \addplot[red, thick, no marks] coordinates {
            (1,1)
            (2,0.95)
            (3,0.89)
            (4,0.88)
            (5,0.88)
            (6,0.81)
            (8,0.81)
            (10,0.75)
            (12,0.71)
            (15,0.57)
            (20,0.47)
            (25,0.5)
            (30,0.42)
        };
        \addplot[violet, dashed, thick, no marks, forget plot] coordinates {
            (1,1)
            (2,0.98)
            (3,0.99)
            (4,0.98)
            (5,0.95)
            (6,0.95)
            (8,0.83)
            (10,0.79)
            (12,0.67)
            (15,0.51)
            (20,0.44)
            (25,0.22)
            (30,0.17)
        };
        \addplot[violet, thick, no marks] coordinates {
            (1,1)
            (2,0.99)
            (3,0.95)
            (4,0.95)
            (5,0.87)
            (6,0.93)
            (8,0.87)
            (10,0.85)
            (12,0.75)
            (15,0.77)
            (20,0.73)
            (25,0.66)
            (30,0.52)
        };

        \legend{Llama3.1-8B, Qwen3-8B, GLM4.5-air};
        
        \end{axis}
    \end{tikzpicture}
    \captionof{figure}{Accuracy rate of single model and agent as the number of required operations grows. Solid: agent; dashed: single model.}
    \label{fig:acc}
    \end{minipage}
    \hfill
    \begin{minipage}[b]{.49\textwidth}
    \centering
    \begin{tabular*}{.99\linewidth}{@{\hspace{5pt}\extracolsep{\fill}} l c c @{\hspace{5pt}}}
        \toprule
        Failure Reason & Single & Agent\\
        \midrule
        Protocol & 11 & 90 \\
        Mathematics & 341 & 126 \\
        \quad Ordering & -- & 60\\
        \quad Arithmetic & -- & 66\\
        \midrule
        Total errors & 352 & 216 \\
        \midrule
        Accuracy & 72.92\% & 83.38\% \\
        \bottomrule
    \end{tabular*}
    \captionof{table}{Failure breakdown for GLM4.5-air in two settings. Protocol failures result from output truncation or invalid format, whereas mathematics failures reflect errors in operator ordering or arithmetic execution.}
    \label{tab:failure}
    \end{minipage}
\end{table}

\subsection{Results}


Figure~\ref{fig:acc} compares performance across expression lengths. 
For the three models, the agent increasingly and substantially outperforms its single-model counterpart as the number of operators grows, demonstrating that, compared to standard CoT alone, leveraging simple tool calls within a loop execution  can effectively enhance the ability to solve long-horizon problems.
At 30 operators, for example, agent accuracy exceeds single-model accuracy by 23, 18, and 35 percentage points for Llama3.1-8B, Qwen3-8B, and GLM4.5-air, respectively.
For short expressions, however, the agent may match or underperform the single model.
The reason is that agent introduces multiple decision and execution steps, and an error at any one of these steps is sufficient to invalidate the final answer.  

Table~\ref{tab:failure} reports the failure analysis specifically for GLM4.5-air.  
Protocol failures denote outputs that reach the generation limit or are in an invalid format, whereas mathematical failures capture either an incorrect operator ordering or erroneous arithmetic execution.
Most importantly, the agent reduces mathematical failures markedly, from 341 under the single-model setting to 126, of which 60 are ordering errors and 66 are arithmetic errors.
Although the agent incurs more protocol failures at the task level, all 90 arise from limitations in experimental conditions or the insufficient ability of the model to strictly follow instructions.
\vspace{-2mm}
\section{Conclusions}
\vspace{-2mm}
This paper investigates the Turing‑completeness of transformers and agents under unbounded‑length inputs. Prior results show that transformers can memorize Turing machines for bounded‑length inputs. We show that finite‑precision transformers suffer from precision errors and cannot memorize non‑converging Turing machines such as Arithmetic. Even infinite-precision transformers exhibit sensitivity decay: input perturbations are diluted for long sequences, making them unable to memorize non‑tame Turing machines under positive output confidence. Random training further yields a zero probability of obtaining such a transformer with limited chain‑of‑thought. By contrast, we show that agents with decision, execution, and memory modules are Turing Complete even with finite precision transformers. At last, we empirically validate these theoretical results.

{\bf Limitations.} 
One of the unsolved problems in this paper is whether the infinite precision transformer is Turing complete without any assumptions. Another future research method is how to design the optimal agent structure so that it has the strongest computing power when the scales of the transformers in it are fixed.






\bibliography{main}
\bibliographystyle{iclr2027_conference}

\newpage
\appendix

\section{Detailed Notations}
\label{app-N}
Appendix A clarifies the main definitions, assumptions, and concepts used in the paper. 
\subsection{Complete Turing machine}
Firstly, we define a Turing machine as follows:

\begin{definition}
 A Turing machine $\mathbb{T}$ is a tuple
\[
\mathbb{T} = (Q,\Sigma_1,\Sigma_2,\delta)
\quad\hbox{or more detailed}\quad
\mathbb{T} = (Q,\Sigma_1,\Sigma_2,\delta,q_0,q_s,B)
\]
defined subject to the following conditions:

(a1): $Q$ is a finite set of states containing two distinguished states $q_0,q_s\in Q$.
The symbol $q_0$ is the initial state, and $q_s$ is the unique halt state.

(a2): $\Sigma_1$ is a finite input alphabet, and $B\notin\Sigma_1$, where $B$ denotes the blank tape symbol.

(a3): $\Sigma_2$ is the tape alphabet satisfying $\Sigma_1\cup\{B\}\subseteq \Sigma_2$.

(a4): The transition function is a transfer function
    \[
    \delta \colon (Q\setminus\{q_s\}) \times \Sigma_2 \to Q \times \Sigma_2 \times \{-1,0,1\}.
    \]

A \textit{configuration} of $\mathbb{T}$ is a triple $(q_p,\sigma,w_p)$ where $q_p\in Q$ is the current state,
$\sigma\colon \mathbb{Z}\to \Sigma_2$ is the tape function, assigning to each integer tape position a symbol from $\Sigma_2$,
and  $w_p\in\mathbb{Z}$ is the current position of the read/write head.

The \emph{computation rules for $\mathbb{T}$:}

(b1): \textit{Initial configuration.}
    Given an input string $x=x_0x_1\cdots x_{n-1}\in\Sigma_1^*$, the initial tape satisfies
    \[
    \sigma(k)=
    \begin{cases}
        x_k & 0\le k\le n-1,\\
        B & \text{otherwise}.
    \end{cases}
    \]
    The initial configuration sets $q_p=q_0$, $w_p=0$.

(b2):  \textit{Single transition step.}
    Suppose the current configuration satisfies $q_p\neq q_s$.
    Let $\sigma(w_p)$ be the symbol at head position $w_p$. Compute
    \[
    (q_\delta,\sigma_\delta,s_\delta) = \delta\big(q_p,\sigma(w_p)\big).
    \]
    The machine transitions to a new configuration by updating:
    \[
    q_p := q_\delta,\quad \sigma(w_p) := \sigma_\delta,\quad w_p := w_p + s_\delta.
    \]

(b3):  \textit{Halt condition and output.}
    If $q_p = q_s$, the machine halts.
    Let $L = \min\{k\in\mathbb{Z}\mid \sigma(k)\neq B\}$ and $R = \max\{k\in\mathbb{Z}\mid \sigma(k)\neq B\}$.
    The output string is $\sigma(L)\sigma(L+1)\cdots\sigma(R)$.
    If no such $L,R$ exists (all tape symbols are blank), the output is the empty string.
\end{definition}

\begin{remark}
Without loss of generality, in this paper, we assume that the final output of the Turing machine does not contain the symbol $B$.
\end{remark}
\subsection{Complete transformer structure}
We consider a fixed-architecture soft-attention transformer. The model consists of an embedding layer, several transformer blocks, and an output layer.

\textbf{Embedding layer $E_{\rm{emb}}$.}
The transformer first maps each symbol $\sigma_i\in\Sigma$ to a vector $v_i\in\mathbb R^d$, where $d$ is the embedding dimension. For an input sequence 
$x=(\sigma_{i_1},\sigma_{i_2},\dots,\sigma_{i_n})\in\Sigma^*$ with length $n$, its embedded representation is written as $x^{(0)}=(v_{i_1},v_{i_2},\dots,v_{i_n})^\top\in\mathbb R^{n\times d}$. 
When no confusion arises, we do not distinguish between the symbolic sequence $x$ and its embedded matrix representation and use the same notation for both.

\textbf{Multi-head attention.}
For an input $x\in\mathbb R^{n\times d}$, the output of the $h$-th attention head is defined by
\[
\ATT_h(x)=\softmax(R(xQ_hK_hx^\top)+M)\,xV_h,
\]
where $Q_h\in\mathbb R^{d\times W}$, 
$K_h\in\mathbb R^{W\times d}$, $V_h\in\mathbb R^{d\times d}$. Without loss of generality, we  take $W=d$ in this paper. 
$M\in\{-\infty,0\}^{n\times n}$ 
is the causal mask satisfying $M_{i,j}=-\infty\quad\Longleftrightarrow\quad j>i$. 

Relative position embedding $R(xQ_hK_hx^\tau)\in\R^{n\times n}$ is used,  where the $(u,v)$-th weight of $R(xQ_hK_hx^\tau)$ is $x[u]Q_hA_{u-v}K_hx^\tau[v]$, $x[i]$ is the $i$-th row of $x$, and $A_{u-v}\in[-1,1]^{W\times W}$ is the matrix for the relative position embedding. Note that $A_{u-v}$ does not depend on which head or layer is used. 
Summing over all $H$ heads gives the total multi-head attention output:
\[
\ATT(x)=\sum_{h=1}^H \ATT_h(x).
\]

\textbf{Position Embedding.}

For position encoding, 
according to the commonly used form of RoPE, for the transformer with width $W$, the position encoding with parameters $\{C_i\}_{i=1}^{[(W+1)/2]}$ is defined as: 
$$(A_{n-m})_{i,i}=\cos((n-m)C_{[(i+1)/2]}),\ i\in[W]$$ 
$$(A_{n-m})_{2i-1,2i}=\sin((n-m)C_{i}),\ i\in[[\frac{W+1}{2}]]$$ 
$$(A_{n-m})_{2i,2i-1}=-\sin((n-m)C_{i}),\ i\in [[\frac{W+1}{2}]]$$
$$(A_{n-m})_{i,j}=0,\ |i-j|\ge 2.$$

This ensures that the position encoding does not tend to infinity as the position increases. And we see $\{C_i\}_{i=1}^{[(W+1)/2]}$ as adjustable parameters in this paper. 

\textbf{Feedforward network.}
For $x\in\mathbb R^{n\times d}$, the feedforward layer is defined as
\[
\FNN(x)=\Relu(xE_1\oplus b)E_2,
\]
where $E_1\in\mathbb R^{d\times W}$, $b\in\mathbb R^{1\times W}$, $E_2\in\mathbb R^{W\times d}$ are parameters, and $xE_1\oplus b$ means that the bias vector $b$ is added to every row of $xE_1$.

\textbf{Transformer block.}
Let $x^{\ell-1}$
be the input to layer $\ell$, and let $x^0=x$. Then we have that: 

\[x^{\ell}=x^{\ell-1}+\ATT_{\ell}(x^{\ell-1})+\FNN_{\ell}(x^{\ell-1}+\ATT_{\ell}(x^{\ell-1})).\]

\textbf{Output layer $E_{\rm{out}}$.}
Assume the transformer has $L$ hidden layers. The output layer applies a linear transformation only to the final position of the last hidden layer:
\[
\mathcal F(x)=x^{L}[k]W_0+c\in\mathbb R^{T+1},
\]
where $k=\len(x^L)$, $W_0\in\mathbb R^{d\times (T+1)}, c\in\mathbb R^{1\times (T+1)}$.

\textbf{The Output of the autoregressive transformer.}
The final output $y=\widehat{\F}(x)$ is obtained as follows, with the initial value $y=()$:


(1) The next token is $\sigma_s\in\Sigma$, where $s=\argmax_{i\in[T+1]}\F_i(x\oplus y)$ and $\F_i(x)$ is the $i$-th weight of $\F(x)$. Let $\sigma_{T+1}=\sigma_0$. 

(2) If $\sigma_s\ne \sigma_0$, then let $y=y\oplus\{\sigma_s\}$ and return to step (1).

(3) If $\sigma_s=\sigma_0$, stop and return the output $y$. For convenience,  we say $y=\widehat{\F}(x)\in\Sigma^*$.

\subsection{Memorization, Simulation, and Turing Completeness}

\paragraph{Two forms of correspondence between transformers and Turing machines.}

Some previous work has attempted to study the power of Transformers in simulating Turing machines, whereas we consider the power of Transformers in memorizing Turing machines. Simulating Turing machines aims to reproduce the calculation process of Turing machines step-by-step, hence establishing computational equivalence. In contrast, memory Turing machines only require the final output to match that of Turing machines, without requiring the process to completely reproduce the internal computation. In computability theory, Turing completeness can be achieved by either simulating or memorizing a general Turing machine. Both approaches establish that the model is equivalent to a Turing machine in computing power.

Statements about the ``Turing completeness'' of transformers in the existing literature can correspond to two different levels of construction.

\medskip
\noindent
\textbf{(1) Family-level correspondence.}
Given a Turing machine $T$, for each input length $n$, one constructs a transformer
\[
\mathcal F_n
\]
that depends on $n$ and can memorize the computation of $T$ on all inputs of length $n$. This gives a model family
\[
\{\mathcal F_n\}_{n\ge1}.
\]
In this sense, the correspondence is
\[
\text{input length } n \quad \longmapsto \quad \text{model } \mathcal F_n,
\]
meaning that different input lengths are allowed to use different transformers.

\medskip
\noindent
\textbf{(2) Single-model correspondence.}
A stronger requirement asks whether there exists one fixed transformer
\(
\F
\)
such that, for inputs of all lengths to the Turing machine $T$, the same $\mathcal F$ can uniformly memorize the computation. In this sense, the correspondence is no longer between a length $n$ and a model $\mathcal F_n$, but rather
\[
\text{Turing machine } T \quad \longmapsto \quad \text{one fixed model } \mathcal F.
\]
The same model must cover the entire input-size range of $T$.

\medskip
\noindent
The key difference is, therefore, that the first formulation allows the model to vary with input length and studies a family of transformer models. The second keeps the model fixed and studies the uniform processing ability of one transformer over arbitrary-length inputs. The first captures length-by-length simulation, whereas the second captures uniform simulation by a single model.

\subsection{The Arithmetic Task}

{\bf Details for Example 1.} Let 
$\Sigma_1=\{0,1,2,3,4,5,6,7,8,9,+,-,*,/,(,),|\}$, 
$\Sigma_2=\Sigma_1\cup\{=,B\}$, and
$\TM_{\Ar}=(Q,\Sigma_1,\Sigma_2,\delta)$ be the Turing machines for evaluating arithmetic expressions over rational number, with $n|m$ to represent a rational number $\frac{n}{m}$. 
The arithmetic problem $\Ar$ is the set $S_\Ar=\{(x, y)\}\subset\Sigma_1^*\times\Sigma_1^*$, where $x$ is an arithmetic formula such as $2+2*(5-1)+3$, and $y$ is the answer to $x$, such as $13$.

We say a Turing machine $\TM_{\Ar}=(Q,\Sigma_1,\Sigma_2,\delta)$ is the Turing machine to solve four arithmetic problems if $\TM_{\Ar}(x)=y$ for any $(x,y)\in S_\Ar$, and $\TM_{\Ar}(x)$ does not halt if $x\in\Sigma_1^*$ is not an arithmetic formula. 

Let $\sigma_T$ be `$=$'. For a transformer $\F$, if $\widehat{\F}(x)$ is `$=2+2*4+3=2+8+3=10+3=13$' and then output $\sigma_0$ to halt. Then $\widehat{\F}_{\out}(x)$ is $13$, which is the sequence  between the last $=$ and $\sigma_0$, other parts are CoT. If $\hat\F_{\rm{out}}(x)=y$ for any $(x,y)\in S_\Ar$, then such a transformer $\F$ can memorize $\TM_{\Ar}$. 

\begin{remark}
It is reasonable to require the Turing machine $\TM_{\Ar}$ not to stop on equations that do not conform to the calculation rules (such as $1+2**6+$); we can also require the Turing machine to output a symbol representing rejection for such equations. The definition of such inputs is not crucial and does not affect the following conclusions.
\end{remark}

\section{Recent results on memorizing Turing machines when inputs are restricted to a fixed length}
\label{old}

A Turing machine $\TM=(Q,\Sigma_1,\Sigma_2,\delta)$ is said to have input length $n$ if every input it receives has a length of at most $n$.
For example, arithmetic over $\Z_p$ on inputs of length $n$ (denoted $\Arpn$) can be implemented by a Turing machine with input length $n$.

Since there are finitely many inputs for a Turing machine $\TM$ with input length $n$, we may define its maximal running time:
$$t_\TM=\max_{x {\rm\, is\ stoppable}}\{|t_{\TM}(x)|\}.$$

A number of prior studies~\citep{Attention-TC,merrill2023expressive,yu2025analyzing,yx} have demonstrated that a CoT transformer is capable of simulating such a Turing machine.
The size of the memorization transformer and the length of the CoT can be divided into two separate groups:

\textbf{Variant-precision transformer \citep{yx}:} 
Let $\TM=(Q,\Sigma_1,\Sigma_2,\delta)$ be a Turing machine with input length $n$. 
Then there exists a transformer $\F$ with a token set $\Sigma_0$ such that
$\Sigma_0\supset\Sigma_2$ and $|\Sigma_0|>|\Sigma_2|+O(|Q|\cdot|\Sigma_2|)$,
with depth $O(\ln t_{\TM})$, width $O(\max\{|Q|\cdot|\Sigma_2|,\ln t_{\TM}\})$, $O(1)$ heads, and precision $O(\ln \ln T_\TM)$,
which is capable of memorizing $\TM$ using a CoT whose length is $O(t_{\TM}(x))$.

\textbf{Fix-precision transformer \citep{yu2025analyzing}:} 
Let $\TM=(Q,\Sigma_1,\Sigma_2,\delta)$ be a Turing machine with input length $n$. 
Then, for any $q \ge 2$, there exists a transformer $\F_q$ with a token set $\Sigma_0$ satisfying $\Sigma_0 \supset \Sigma_2$ and $|\Sigma_0| \ge |\Sigma_2| + 1$, with depth $O(T^{O(t_{\TM})})$ and $O(T)$ heads, which is capable of memorizing $\TM$ using a CoT of length $O(1)$.


Taken together, these two results indicate that a Turing machine with inputs of bounded length can be memorized by a transformer, even when only a short CoT is used. 
In practice, most Turing machines are not subject to such bounds. Consequently, in the main paper, we focus mainly on the more common setting of general Turing machines.

\section{Incompleteness of Finite-Precision Transformers}
\label{app:proof_thm3}
In this section, we will prove Theorem~\ref{zzdd1}.
\subsection{Detailed Definition for Finite-Precision Transformers}
\label{precision}

We first provide a precise definition of finite-precision computation in a
Transformer.

\begin{definition}
For $x\in\mathbb{R}$ with $|x|\leq 10^q$, let $[x]_q$ denote the value
obtained by retaining $q$ decimal places of $x$. If $|x|>10^q$, define
\[
[x]_q=\mathrm{Inf}.
\]

For a vector or matrix $M$, define $[M]_q$ entrywise by
\[
([M]_q)_{i,j}=[M_{i,j}]_q.
\]
\end{definition}

We now define the $q$-precision version of a Transformer.

\begin{definition}
\label{def:finite-precision}
Let $\mathcal{F}$ be a Transformer. Its $q$-precision version, denoted by
$\mathcal{F}_q$, is defined as follows.

First, $\mathcal{F}$ and $\mathcal{F}_q$ have the same depth $D$, width
$W$, number of heads $H$, and vocabulary $\Sigma$. Every parameter $p$
of $\mathcal{F}$ is replaced by $[p]_q$ in $\mathcal{F}_q$.

Second, all arithmetic operations in the positional embedding, hidden
layers, and output layer are carried out under $q$-precision.



For the trigonometric functions in the positional encoding, we maintain
the reduced phase recursively so that an unbounded positional product
need not be formed. Let
\[
M_q:=2[\pi]_q.
\]
For each $q$-precision frequency $\omega$, let
\[
\bar{\omega}
=
\omega-M_q\left\lfloor\frac{\omega}{M_q}\right\rfloor
\in[0,M_q).
\]
Define $r_0=0$ and, for $n\geq 0$,
\[
r_{n+1}
=
\begin{cases}
r_n+\bar{\omega},
& r_n<M_q-\bar{\omega},\\[1mm]
r_n-(M_q-\bar{\omega}),
& r_n\geq M_q-\bar{\omega}.
\end{cases}
\]
Thus $0\leq r_n<M_q$ for every $n$, and the possibly unbounded
product $n\omega$ is never formed during the positional-encoding
calculation.

We then define
\[
\sin_q(n\omega)=[\sin(r_n)]_q,
\qquad
\cos_q(n\omega)=[\cos(r_n)]_q.
\]

For vector addition, define
\begin{equation}
[x+z]_q
=
\bigl([x_i+z_i]_q\bigr)_i.
\label{eq:q-precision-addition}
\end{equation}

For the inner product, every intermediate multiplication is also rounded:
\begin{equation}
\langle x,z\rangle_q
=
\left[
\sum_{i=1}^{M}
[x_i z_i]_q
\right]_q.
\label{eq:q-precision-inner-product}
\end{equation}

For softmax, let
\begin{equation}
x_{\max}
=
\max_{1\leq j\leq M}x_j.
\label{eq:q-precision-softmax-max}
\end{equation}

Define
\begin{equation}
E_i
=
\left[
e^{[x_i-x_{\max}]_q}
\right]_q
\label{eq:q-precision-softmax-exp}
\end{equation}
and
\begin{equation}
Z
=
\left[
\sum_{j=1}^{M}
E_j
\right]_q.
\label{eq:q-precision-softmax-normalizer}
\end{equation}

Then
\begin{equation}
\operatorname{softmax}_q(x)_i
=
\left[
\frac{E_i}{Z}
\right]_q.
\label{eq:q-precision-softmax}
\end{equation}

Finally, values outside the precision range are handled according to
\begin{equation}
\frac{c}{\mathrm{Inf}}
=
0,
\qquad
c+\mathrm{Inf}
=
\mathrm{Inf}
\qquad
when\ c\neq\mathrm{Inf},
\label{eq:q-precision-inf-rule-1}
\end{equation}
\begin{equation}
c_{\ne0}\cdot\mathrm{Inf}
=
\mathrm{Inf},
\qquad
\frac{\mathrm{Inf}}{0}
=
\mathrm{Inf},
\label{eq:q-precision-inf-rule-2}
\end{equation}
and
\begin{equation}
\frac{c}{0}
=
\frac{\mathrm{Inf}}{\mathrm{Inf}}
=
\mathrm{Inf}+\mathrm{Inf}
=
\mathrm{Inf}-\mathrm{Inf}
=
0\cdot\mathrm{Inf}
=
\mathrm{NaN}.
\label{eq:q-precision-nan-rule}
\end{equation}

If $\mathrm{NaN}$ appears at any intermediate step, the final output is
also defined to be $\mathrm{NaN}$. Without loss of generality, we require
that neither $\mathrm{Inf}$ nor $\mathrm{NaN}$ appears in the final output of the transformer in any results in this paper.
\end{definition}

\subsection{Proof of Theorem~\ref{zzdd1}}

We begin with a periodicity property of finite-precision positional
encoding.

\begin{lemma}
\label{xhj}
For every $q$-precision Transformer $\mathcal{F}_q$, there exists a
positive integer $C$ such that
\begin{equation}
A_i
=
A_{i+C}
\label{eq:position-periodicity}
\end{equation}
for every position $i$. Thus, under finite precision, the positional
encoding is periodic.
\end{lemma}

\begin{proof}
After $q$-precision truncation, every frequency parameter appearing in
the positional encoding is rational. Let the corresponding reduced
frequencies be
\[
\bar{\omega}_1,\ldots,\bar{\omega}_R,
\]
and let
\[
M_q=2[\pi]_q.
\]
Since all these numbers have at most $q$ decimal places, for each $r$
there exist integers $a_r$ and $b>0$ such that
\[
\bar{\omega}_r=\frac{a_r}{10^q},
\qquad
M_q=\frac{b}{10^q}.
\]

Under the recursive phase update in Definition~\ref{def:finite-precision},
the phase corresponding to $\bar{\omega}_r$ is therefore
\[
r_n^{(r)}
=
\frac{n a_r \bmod b}{10^q}.
\]
Hence this phase is periodic, with a period
\[
C_r=\frac{b}{\gcd(a_r,b)}.
\]
Let
\[
C=\operatorname{lcm}(C_1,\ldots,C_R).
\]
Then, for every $n$ and every $r$,
\[
r_{n+C}^{(r)}=r_n^{(r)}.
\]
Therefore all sine and cosine components of the RoPE positional
encoding agree at positions $n$ and $n+C$, and hence
\[
A_n=A_{n+C}.
\]
\end{proof}

We can now prove Theorem~\ref{zzdd1}.

\begin{proof}
For a $q$-precision Transformer $\mathcal{F}_q$, let
\[
\mathcal{F}_q(s)_{m,n}
\]
denote the hidden state at position $n$ in the $m$-th hidden layer.

Let $\mathcal{T}$ be a non-converging Turing machine.

For $x,u,z\in\Sigma^*$ and $k\in\mathbb{N}_+$, define
\begin{equation}
s_{x,u,k,z}
=
x\oplus u^k\oplus z.
\label{eq:periodic-input-definition}
\end{equation}
We shall repeatedly use causality: if two input sequences have the same
prefix up to position $r$, then their hidden states at every position
$i\le r$ are identical at every hidden layer.
Assume, toward a contradiction, that there exist $q\in\mathbb{Z}_+$ and
a $q$-precision Transformer $\mathcal{F}_q$ that memorizes
$\mathcal{T}$. We study the values of $\mathcal{F}_q$ on the family
$s_{x,u,k,z}$.

\medskip
\noindent\textbf{Part One.}

Fix $x,u\in\Sigma^*$ with $u\neq\varnothing$, an integer $k\ge C+1$,
a continuation $w\in\Sigma^*$, and a hidden-layer index $m$. Let
\[
P=C|u|,
\]
and define
\[
s^-:=s_{x,u,k-C,w},
\qquad
s^+:=s_{x,u,k,w},
\]
with
\[
N_-:=|s^-|,
\qquad
N_+:=|s^+|=N_-+P.
\]

Suppose that the following three conditions hold.

First, the final hidden states agree:
\begin{equation}
\mathcal{F}_q(s^-)_{m,N_-}
=
\mathcal{F}_q(s^+)_{m,N_+}.
\tag{c1}
\label{eq:c1}
\end{equation}

Second, the hidden states corresponding to the common suffix $w$ agree:
\begin{equation}
\mathcal{F}_q(s^-)_{m,|x|+(k-C)|u|+t}
=
\mathcal{F}_q(s^+)_{m,|x|+k|u|+t}
\tag{c2}
\label{eq:c2}
\end{equation}
for every
\[
1\le t\le |w|.
\]

Third, for every additional position
\[
i\in
\{
|x|+(k-C)|u|+1,\ldots,|x|+k|u|
\},
\]
there exist more than $10^{2q}$ distinct indices
\[
j\le |x|+(k-C)|u|
\]
such that
\begin{equation}
\mathcal{F}_q(s^+)_{m,i}
=
\mathcal{F}_q(s^-)_{m,j},
\qquad
A_{N_+-i}
=
A_{N_--j}.
\tag{c3}
\label{eq:c3}
\end{equation}

We claim that
\begin{equation}
\mathcal{F}_q(s^-)_{m+1,N_-}
=
\mathcal{F}_q(s^+)_{m+1,N_+}.
\label{eq:part-one-conclusion}
\end{equation}

For any input $s$, the update at its final position has the form
\begin{align}
\mathcal{F}_q(s)_{m+1,|s|}
&=
\mathcal{F}_q(s)_{m,|s|}
+
\operatorname{Att}(\mathcal{F}_q(s)_m)[|s|]
\nonumber\\
&\quad+
\operatorname{FNN}\!\left(
\mathcal{F}_q(s)_{m,|s|}
+
\operatorname{Att}(\mathcal{F}_q(s)_m)[|s|]
\right).
\label{eq:transformer-layer-update}
\end{align}

By condition~\eqref{eq:c1}, the residual inputs at the two final positions
are identical. It therefore remains to compare the two attention outputs.

Fix an attention head $h$. Write
\[
h_i^-=\mathcal{F}_q(s^-)_{m,i},
\qquad
h_i^+=\mathcal{F}_q(s^+)_{m,i},
\]
and define
\begin{equation}
\lambda_i^{(\pm,h)}
=
\left[
h_{N_\pm}^{\pm}
Q_h
A_{N_\pm-i}
K_h
(h_i^\pm)^\top
\right]_q.
\label{eq:attention-logit-notation}
\end{equation}

Every position of $s^-$ has a corresponding position in $s^+$. For
positions in the common prefix $x\oplus u^{k-C}$, the hidden states agree
by causality, while the two relative position indices differ by
\[
P=C|u|.
\]
Since $A_r=A_{r+C}$, we also have
\[
A_{r+P}=A_r.
\]
For positions in the suffix $w$, the hidden states agree by
condition~\eqref{eq:c2}, and the relative position indices are identical.
Together with~\eqref{eq:c1}, the corresponding logits and value vectors
are therefore equal.

Moreover, by condition~\eqref{eq:c3}, every additional position in $s^+$
has more than $10^{2q}$ matching positions in $s^-$ with the same hidden
state and the same relative positional matrix. Hence its attention logit
also agrees with the logits at those matching positions.

Thus the maximum logit is the same in the two computations. Denote it by
$W_h$, and define
\[
E_i^{(\pm,h)}
=
\left[
e^{[\lambda_i^{(\pm,h)}-W_h]_q}
\right]_q,
\]
and
\[
Z_\pm^{(h)}
=
\left[
\sum_{i=1}^{N_\pm}
E_i^{(\pm,h)}
\right]_q.
\]

We now consider the additional positions in $s^+$.

If
\[
E_i^{(+,h)}=0
\]
for every additional position $i$, then these positions contribute zero.
All remaining corresponding logits and value vectors agree, and hence the
two attention-head outputs are equal.

Otherwise, there exists an additional position $i$ such that
\[
E_i^{(+,h)}>0.
\]
By condition~\eqref{eq:c3}, there are more than $10^{2q}$ positions in
$s^-$ with the same positive exponential value. Since every positive
$q$-precision number is at least $10^{-q}$,
\[
\sum_{j=1}^{N_-}E_j^{(-,h)}
>
10^{2q}10^{-q}
=
10^q.
\]
Therefore,
\[
Z_-^{(h)}=\mathrm{Inf}.
\]
The same matching positions also occur in $s^+$, so
\[
Z_+^{(h)}=\mathrm{Inf}.
\]
By the finite-precision convention
\[
\frac{c}{\mathrm{Inf}}=0,
\]
all attention weights of this head vanish in both computations. Hence the
two attention-head outputs are again equal.

The same argument applies to every attention head. Combining this with
condition~\eqref{eq:c1} and
\eqref{eq:transformer-layer-update} proves
\eqref{eq:part-one-conclusion}.

\medskip
\noindent\textbf{Part Two.}

We next prove a layerwise relation that is uniform over all continuations.
Let
\[
R_q:=10^{2q}+1.
\]

We prove by induction on $m$ that there exists an integer $K_m\ge C+1$
such that, for every $k\ge K_m$ and every $w\in\Sigma^*$,
\begin{equation}
\mathcal{F}_q(s_{x,u,k-C,w})_{m,|s_{x,u,k-C,w}|}
=
\mathcal{F}_q(s_{x,u,k,w})_{m,|s_{x,u,k,w}|}.
\tag{$\mathcal I_m$}
\label{eq:layerwise-invariant}
\end{equation}

For $m=0$, the final token is identical in the two sequences. If
$w\neq\varnothing$, it is the last token of $w$; if $w=\varnothing$,
it is the last token of $u$. Hence~\eqref{eq:layerwise-invariant} holds,
for example with
\[
K_0=C+1.
\]

Suppose that~\eqref{eq:layerwise-invariant} holds at the $m$-th hidden
layer. Define
\begin{equation}
K_{m+1}
=
K_m+C(R_q+1).
\label{eq:next-threshold}
\end{equation}

Fix $k\ge K_{m+1}$ and $w\in\Sigma^*$. We verify the three conditions in
Part One.

Condition~\eqref{eq:c1} is exactly
\eqref{eq:layerwise-invariant} applied to the continuation $w$.

For each $1\le t\le |w|$, apply
\eqref{eq:layerwise-invariant} to the continuation $w[1:t]$.
By causality, this gives condition~\eqref{eq:c2}.

It remains to verify condition~\eqref{eq:c3}. Let
\[
i\in
\{
|x|+(k-C)|u|+1,\ldots,|x|+k|u|
\}.
\]
There are unique
\[
b\in\{k-C+1,\ldots,k\},
\qquad
t\in\{1,\ldots,|u|\}
\]
such that
\[
i
=
|x|+(b-1)|u|+t.
\]

Let
\[
P=C|u|.
\]
For
\[
r=1,\ldots,R_q,
\]
define
\[
j_r=i-rP.
\]
Because
\[
k\ge K_m+C(R_q+1),
\]
all the indices $j_r$ lie inside the common repeated $u$-block, and all
applications of~\eqref{eq:layerwise-invariant} below involve an exponent
at least $K_m$.

Applying~\eqref{eq:layerwise-invariant} successively to the continuation
$u[1:t]$, and using causality, gives
\[
\mathcal{F}_q(s_{x,u,k,w})_{m,i}
=
\mathcal{F}_q(s_{x,u,k-C,w})_{m,j_r},
\qquad
1\le r\le R_q.
\]

Moreover,
\[
|s_{x,u,k-C,w}|-j_r
=
|s_{x,u,k,w}|-i+(r-1)P.
\]
Since
\[
P=C|u|
\]
is a multiple of the RoPE period $C$, Lemma~\ref{xhj} gives
\[
A_{|s_{x,u,k-C,w}|-j_r}
=
A_{|s_{x,u,k,w}|-i}.
\]

The indices $j_r$ are distinct and
\[
R_q=10^{2q}+1>10^{2q}.
\]
Hence condition~\eqref{eq:c3} holds.

Part One therefore gives
\[
\mathcal{F}_q(s_{x,u,k-C,w})_{m+1,|s_{x,u,k-C,w}|}
=
\mathcal{F}_q(s_{x,u,k,w})_{m+1,|s_{x,u,k,w}|},
\]
which is exactly~\eqref{eq:layerwise-invariant} at layer $m+1$.

This completes the induction.

\medskip
\noindent\textbf{Part Three.}

We now derive the contradiction.

Since $\mathcal{T}$ is non-converging, choose witnesses
$x,u,z\in\Sigma_1^*$ from Definition~3. Necessarily,
\[
u\neq\varnothing.
\]

Let $L$ be the depth of $\mathcal{F}_q$, and choose any
\[
k\ge K_L.
\]

By Part Two, for every continuation $v\in\Sigma^*$,
\begin{align}
&
\mathcal{F}_q(s_{x,u,k-C,z\oplus v})_
{L,|s_{x,u,k-C,z\oplus v}|}
\nonumber\\
&\qquad=
\mathcal{F}_q(s_{x,u,k,z\oplus v})_
{L,|s_{x,u,k,z\oplus v}|}.
\label{eq:final-layer-continuation}
\end{align}

Applying the common output layer gives
\begin{equation}
\mathcal{F}_q(s_{x,u,k-C,z\oplus v})
=
\mathcal{F}_q(s_{x,u,k,z\oplus v})
\label{eq:autoregressive-step-equality}
\end{equation}
for every $v\in\Sigma^*$.

Taking $v=\varnothing$, the first generated tokens are identical.
Suppose inductively that the two generations have produced the same prefix
$v$. Then~\eqref{eq:autoregressive-step-equality}, applied to this same
$v$, shows that the next generated tokens are also identical. The same
argument applies to the EOS decision. Hence
\begin{equation}
\widehat{\mathcal{F}}_q(s_{x,u,k-C,z})
=
\widehat{\mathcal{F}}_q(s_{x,u,k,z}).
\label{eq:complete-output-equality}
\end{equation}

Therefore, their extracted final outputs are identical.

However, by Definition~3, both
\[
s_{x,u,k-C,z}
\qquad\text{and}\qquad
s_{x,u,k,z}
\]
are stoppable inputs and
\[
\mathcal{T}(s_{x,u,k-C,z})
\neq
\mathcal{T}(s_{x,u,k,z}).
\]
This contradicts the assumption that $\mathcal{F}_q$ memorizes
$\mathcal{T}$. Therefore, no $q$-precision Transformer can memorize a
non-converging Turing machine, and the theorem follows.
\end{proof}

\section{Sensitivity Bounds for Infinite-Precision Transformers}
\label{app:proof_thm4}
In this section, we will prove Theorems \ref{th-1} and \ref{th-2}.

\subsection{Lipschitz Continuity of Basic Maps}

This section proves the Lipschitz continuity, on relevant bounded regions, of several basic component maps in a standard pre-norm transformer, which will be used in Section \ref{app-th4}.

\paragraph{Notation.}
Unless otherwise stated, vector norms are Euclidean norms $\|\cdot\|_2$, and matrix norms are induced operator norms, still denoted by $\|\cdot\|$. For a vector-valued map $f$, if there exists a constant $L>0$ such that
\[
\|f(x)-f(y)\|_2 \le L\|x-y\|_2
\]
for all $x,y$ in its domain, then $f$ is said to be $L$-Lipschitz on that region.

\subsubsection{Lipschitz Continuity under Finite Composition}

\begin{lemma}[Composition and addition preserve Lipschitz continuity]
\label{composition}
\leavevmode
\begin{itemize}
    \item If $f:X\to Y$ is $L_f$-Lipschitz and $g:Y\to Z$ is $L_g$-Lipschitz, then $g\circ f$ is $L_gL_f$-Lipschitz.
    \item If $f,g:X\to \mathbb R^d$ are $L_f$- and $L_g$-Lipschitz respectively, then $f+g$ is $(L_f+L_g)$-Lipschitz.
\end{itemize}
\end{lemma}

\begin{proof}
The first statement follows from
\[
\|g(f(x))-g(f(y))\|
\le
L_g\|f(x)-f(y)\|
\le
L_gL_f\|x-y\|.
\]
The second follows from
\[
\|(f+g)(x)-(f+g)(y)\|
\le
\|f(x)-f(y)\|+\|g(x)-g(y)\|
\le
(L_f+L_g)\|x-y\|.
\]
\end{proof}

\subsubsection{Affine Maps, ReLU, and Addition}

\begin{lemma}[Lipschitz continuity of affine maps]
Let $A\in\mathbb R^{m\times d}$, $b\in\mathbb R^m$, and $T(u)=Au+b$.

Then $T$ is globally Lipschitz, with a Lipschitz constant of at most $\|A\|$.
\end{lemma}

\begin{proof}
For any $u,v\in\mathbb R^d$, $T(u)-T(v)=A(u-v)$.
Therefore
\[
\|T(u)-T(v)\|_2
=\|A(u-v)\|_2
\le \|A\|\,\|u-v\|_2.
\]
Thus $T$ is $\|A\|$-Lipschitz.
\end{proof}

\begin{lemma}[Lipschitz continuity of ReLU]
The coordinatewise ReLU map
\[
\Relu(z)=(\max\{z_1,0\},\dots,\max\{z_d,0\})
\]
is globally $1$-Lipschitz.
\end{lemma}

\begin{proof}
For the one-dimensional function $\rho(t)=\max\{t,0\}$, we have $|\rho(a)-\rho(b)|\le |a-b|$
for any $a,b\in\mathbb R$. Summing over coordinates gives
\[
\|\Relu(u)-\Relu(v)\|_2^2
=
\sum_{k=1}^d |\rho(u_k)-\rho(v_k)|^2
\le
\sum_{k=1}^d |u_k-v_k|^2
=
\|u-v\|_2^2.
\]
Hence $\Relu$ is globally $1$-Lipschitz.
\end{proof}

\begin{lemma}[Lipschitz continuity of addition]
Define
\[
S(u_1,\dots,u_m)=u_1+\cdots+u_m,
\qquad u_i\in\mathbb R^d.
\]
Then $S$ is Lipschitz with respect to the product-space variables. More precisely, if the product space is equipped with
\[
\|(u_1,\dots,u_m)\|_{\oplus}:=\sum_{i=1}^m \|u_i\|_2,
\]
then $S$ is $1$-Lipschitz. If it is equipped with the Euclidean product norm
\[
\|(u_1,\dots,u_m)\|_{2,\oplus}:=\Bigl(\sum_{i=1}^m \|u_i\|_2^2\Bigr)^{1/2},
\]
then $S$ is $\sqrt m$-Lipschitz.
\end{lemma}

\begin{proof}
For any $(u_1,\dots,u_m)$ and $(v_1,\dots,v_m)$,
\[
S(u_1,\dots,u_m)-S(v_1,\dots,v_m)=\sum_{i=1}^m (u_i-v_i).
\]
By the triangle inequality,
\[
\|S(u_1,\dots,u_m)-S(v_1,\dots,v_m)\|_2
\le
\sum_{i=1}^m \|u_i-v_i\|_2
=
\|(u_1-v_1,\dots,u_m-v_m)\|_{\oplus}.
\]
Thus $S$ is $1$-Lipschitz under $\|\cdot\|_{\oplus}$.

By Cauchy--Schwarz,
\[
\sum_{i=1}^m \|u_i-v_i\|_2
\le
\sqrt m\Bigl(\sum_{i=1}^m \|u_i-v_i\|_2^2\Bigr)^{1/2}.
\]
Thus $S$ is $\sqrt m$-Lipschitz under $\|\cdot\|_{2,\oplus}$.
\end{proof}

Obviously, the following conclusion can be drawn from the above lemma.
\begin{lemma}[Lipschitz continuity of $\FNN$]
Let $\FNN(x)=\Relu(xE_1+b)E_2$, where $E_1\in R^{d\times w}, E_2\in\mathbb R^{w\times d}, b\in\mathbb R^w$.

Then $\FNN$ is globally Lipschitz, with Lipschitz constant $L_{\FNN}\le \|E_1\|\,\|E_2\|$.
\end{lemma}

Combining the preceding Lipschitz properties of addition and the feed-forward network, we obtain the following result for the position-level update map.

\begin{proposition}[Lipschitz continuity of the position-level update on relevant bounded regions]
Consider the position-level update map of a standard transformer:
\[
f^{(\ell)}(u,b_1,\dots,b_H)
=
u+\sum_{h=1}^H b_h
+
\FNN_\ell\!\left(
u+\sum_{h=1}^H b_h
\right).
\]
If all parameters are uniformly bounded and the input $(u,b_1,\dots,b_H)$ lies in a bounded region, then for every layer $\ell$, $f^{(\ell)}:\mathbb R^d\times(\mathbb R^d)^H\to\mathbb R^d$ is Lipschitz with respect to all input variables.
\end{proposition}

\subsection{Proof of Theorem~\ref{th-1}}
\label{app-th4}
Now we prove Theorem~\ref{th-1}.

\begin{proof}
Let $n=\len(x)$, and take any $x'\in B(x,1)$. If $x'=x$, then
$\|\F(x)-\F(x')\|_2=0$. Otherwise, by the definition of $B(x,1)$, the two
sequences $x$ and $x'$ differ at exactly one position $i<n$. Let $y_j^{(\ell)}$ and $\tilde y_j^{(\ell)}$ denote the hidden states at
position $j$ after layer $\ell$ corresponding to $x$ and $x'$, respectively,
and define $\Delta y_j^{(\ell)}:=y_j^{(\ell)}-\tilde y_j^{(\ell)}$.

Since $\Sigma$ is finite, the embedding vectors are uniformly bounded. Set
$M_0:=\max_{\sigma\in\Sigma}\|v_\sigma\|_2$ and
$D_0:=\max_{\sigma,\tau\in\Sigma}\|v_\sigma-v_\tau\|_2$. Then
$\|\Delta y_i^{(0)}\|_2\le D_0$ and $\Delta y_j^{(0)}=0$ for $j\ne i$.
Moreover, by causality, $\Delta y_j^{(\ell)}=0$ for all $j<i$ and
$0\le\ell\le L$.

We will prove that, for every $1\le\ell\le L$, there exists a constant
$C_\ell>0$, depending only on the fixed transformer $\F$, such that
\[
\|\Delta y_j^{(\ell)}\|_2
\le
C_\ell\frac{(\ln j)^{\ell-1}}{j},
\qquad j>i.
\tag{20}
\]

\medskip
\noindent\textbf{Step 1: Uniform boundedness of hidden states and attention weights.}

For the $h$-th attention head in layer $\ell$, let
\[
b_{j,\ell,h}
=
\sum_{w\le j}\alpha_{j,w}^{(\ell,h)}
y_w^{(\ell-1)}V_{\ell,h},
\qquad
\alpha_{j,w}^{(\ell,h)}
=
\frac{\exp(a_{j,w}^{(\ell,h)})}
{\sum_{t\le j}\exp(a_{j,t}^{(\ell,h)})},
\]
where
\[
a_{j,w}^{(\ell,h)}
=
y_j^{(\ell-1)}
Q_{\ell,h}A_{j-w}K_{\ell,h}
\bigl(y_w^{(\ell-1)}\bigr)^\top.
\]

We first show that the hidden states are uniformly bounded. Since the
attention weights are nonnegative and sum to one,
\[
\|b_{j,\ell,h}\|_2
\le
\|V_{\ell,h}\|
\max_{w\le j}\|y_w^{(\ell-1)}\|_2.
\]
By the Lipschitz continuity of $\FNN_\ell$, if
$\|y_w^{(\ell-1)}\|_2\le M_{\ell-1}$ for every $w$, then there exists
$M_\ell>0$, depending only on $M_{\ell-1}$ and the parameters of layer
$\ell$, such that $\|y_j^{(\ell)}\|_2\le M_\ell$ for every $j$. The same
bound holds for $\tilde y_j^{(\ell)}$. Since the depth $L$ is fixed, there
exists a constant $M>0$, depending only on $\F$, such that
\[
\|y_j^{(\ell)}\|_2\le M,
\qquad
\|\tilde y_j^{(\ell)}\|_2\le M
\tag{21}
\]
for every position $j$ and every $0\le\ell\le L$. In particular,
\[
\|\Delta y_i^{(\ell)}\|_2
\le
\|y_i^{(\ell)}\|_2+\|\tilde y_i^{(\ell)}\|_2
\le 2M.
\tag{22}
\]

Since every entry of $A_r$ is uniformly bounded and the width of the
transformer is fixed, $\sup_{r\in\mathbb Z}\|A_r\|<\infty$. Together with
(21), this implies that there exists $S>0$, depending only on $\F$, such that
\[
|a_{j,w}^{(\ell,h)}|\le S,
\qquad
|\tilde a_{j,w}^{(\ell,h)}|\le S
\tag{23}
\]
for every $j,w,\ell,h$. Therefore,
\[
\alpha_{j,w}^{(\ell,h)}
\le
\frac{e^S}{je^{-S}}
=
\frac{e^{2S}}{j}.
\]
The same bound holds for $\tilde\alpha_{j,w}^{(\ell,h)}$. Hence there exists
$C_0>0$ such that
\[
\alpha_{j,w}^{(\ell,h)}
\le\frac{C_0}{j},
\qquad
\tilde\alpha_{j,w}^{(\ell,h)}
\le\frac{C_0}{j}.
\tag{24}
\]

We next record two perturbation estimates that will be repeatedly used.
By adding and subtracting the corresponding mixed term and using (21), there
exists $L_a>0$, depending only on $\F$, such that
\[
\left|
a_{j,w}^{(\ell,h)}
-
\tilde a_{j,w}^{(\ell,h)}
\right|
\le
L_a
\left(
\|\Delta y_j^{(\ell-1)}\|_2
+
\|\Delta y_w^{(\ell-1)}\|_2
\right).
\tag{25}
\]

For the softmax map, write $\alpha=\softmax(a)$ and
$\tilde\alpha=\softmax(\tilde a)$, where $a,\tilde a\in\mathbb R^j$.
Its Jacobian satisfies
\[
\frac{\partial\alpha_u}{\partial a_t}
=
\alpha_u(\delta_{ut}-\alpha_t).
\]
Hence, for every fixed $t$,
\[
\begin{aligned}
\sum_{u\le j}
\left|
\frac{\partial\alpha_u}{\partial a_t}
\right|
&=
\alpha_t(1-\alpha_t)
+
\sum_{u\ne t}\alpha_u\alpha_t \\
&=
2\alpha_t(1-\alpha_t)
\le
2\alpha_t.
\end{aligned}
\]
Every point on the line segment joining $a$ and $\tilde a$ also has all its
coordinates in $[-S,S]$. Thus (24) holds along this segment. By the integral
form of the mean-value theorem, there exists a constant $C_s>0$ such that
\[
\sum_{u\le j}|\alpha_u-\tilde\alpha_u|
\le
\frac{C_s}{j}
\sum_{t\le j}|a_t-\tilde a_t|.
\tag{26}
\]

\medskip
\noindent\textbf{Step 2: Base case $\ell=1$.}

Fix $j>i$ and an attention head $h$. Write
\[
b_{j,1,h}-\tilde b_{j,1,h}
=
I_{j,1,h}+II_{j,1,h},
\]
where
\[
I_{j,1,h}
:=
\sum_{w\le j}
\alpha_{j,w}^{(1,h)}
\Delta y_w^{(0)}V_{1,h},
\qquad
II_{j,1,h}
:=
\sum_{w\le j}
\bigl(
\alpha_{j,w}^{(1,h)}
-
\tilde\alpha_{j,w}^{(1,h)}
\bigr)
\tilde y_w^{(0)}V_{1,h}.
\]

Since the layer-$0$ perturbation occurs only at position $i$, (24) gives
\[
\begin{aligned}
\|I_{j,1,h}\|_2
&\le
\alpha_{j,i}^{(1,h)}
\|\Delta y_i^{(0)}\|_2
\|V_{1,h}\| \\
&\le
\frac{C_0D_0\|V_{1,h}\|}{j}.
\end{aligned}
\tag{27}
\]

Since $j>i$, we have $\Delta y_j^{(0)}=0$, and only $w=i$ has a nonzero
perturbation. Thus (25) implies
\[
\sum_{w\le j}
\left|
a_{j,w}^{(1,h)}
-
\tilde a_{j,w}^{(1,h)}
\right|
\le
L_aD_0.
\]
By (26),
\[
\sum_{w\le j}
\left|
\alpha_{j,w}^{(1,h)}
-
\tilde\alpha_{j,w}^{(1,h)}
\right|
\le
\frac{C_sL_aD_0}{j}.
\]
Using (21),
\[
\begin{aligned}
\|II_{j,1,h}\|_2
&\le
\sum_{w\le j}
\left|
\alpha_{j,w}^{(1,h)}
-
\tilde\alpha_{j,w}^{(1,h)}
\right|
\|\tilde y_w^{(0)}\|_2
\|V_{1,h}\| \\
&\le
\frac{C}{j}
\end{aligned}
\tag{28}
\]
for some constant $C>0$ depending only on $\F$.

Combining (27) and (28), there exists $C'>0$ such that
\[
\|b_{j,1,h}-\tilde b_{j,1,h}\|_2
\le
\frac{C'}{j}.
\tag{29}
\]

For each layer $\ell$, define the position-wise transformer update
\[
f^{(\ell)}(u,b_1,\ldots,b_H)
:=
u+\sum_{h=1}^H b_h
+
\FNN_\ell
\left(
u+\sum_{h=1}^H b_h
\right).
\]
By the Lipschitz continuity of $\FNN_\ell$, $f^{(\ell)}$ is Lipschitz. Let
$L_{\mathrm{act},\ell}$ be one of its Lipschitz constants. Then
\[
\|\Delta y_j^{(1)}\|_2
\le
L_{\mathrm{act},1}
\left(
\|\Delta y_j^{(0)}\|_2
+
\sum_{h=1}^H
\|b_{j,1,h}-\tilde b_{j,1,h}\|_2
\right).
\]
Since $\Delta y_j^{(0)}=0$ for $j>i$, (29) yields
\[
\|\Delta y_j^{(1)}\|_2
\le
\frac{C_1}{j}
=
C_1\frac{(\ln j)^0}{j}
\]
for some $C_1>0$. Hence (20) holds for $\ell=1$.

\medskip
\noindent\textbf{Step 3: Induction step.}

Suppose $\ell\ge2$ and assume that (20) holds at layer $\ell-1$, namely,
\[
\|\Delta y_w^{(\ell-1)}\|_2
\le
C_{\ell-1}
\frac{(\ln w)^{\ell-2}}{w},
\qquad
w>i.
\tag{30}
\]

Fix $j>i$ and an attention head $h$. Again write
\[
b_{j,\ell,h}-\tilde b_{j,\ell,h}
=
I_{j,\ell,h}+II_{j,\ell,h},
\]
where
\[
I_{j,\ell,h}
:=
\sum_{w\le j}
\alpha_{j,w}^{(\ell,h)}
\Delta y_w^{(\ell-1)}V_{\ell,h},
\qquad
II_{j,\ell,h}
:=
\sum_{w\le j}
\bigl(
\alpha_{j,w}^{(\ell,h)}
-
\tilde\alpha_{j,w}^{(\ell,h)}
\bigr)
\tilde y_w^{(\ell-1)}V_{\ell,h}.
\]

By causality, $\Delta y_w^{(\ell-1)}=0$ for $w<i$. Using (22), (24), and
(30), we obtain
\[
\begin{aligned}
\|I_{j,\ell,h}\|_2
&\le
\frac{C_0\|V_{\ell,h}\|}{j}
\left(
\|\Delta y_i^{(\ell-1)}\|_2
+
\sum_{w=i+1}^{j}
\|\Delta y_w^{(\ell-1)}\|_2
\right) \\
&\le
\frac{C}{j}
\left(
1+
\sum_{w=i+1}^{j}
\frac{(\ln w)^{\ell-2}}{w}
\right).
\end{aligned}
\]
By integral comparison,
\[
\sum_{w=i+1}^{j}
\frac{(\ln w)^{\ell-2}}{w}
=
O\!\left((\ln j)^{\ell-1}\right).
\]
Therefore,
\[
\|I_{j,\ell,h}\|_2
\le
C\frac{(\ln j)^{\ell-1}}{j}.
\tag{31}
\]

For the second term, (25) gives
\[
\sum_{w\le j}
\left|
a_{j,w}^{(\ell,h)}
-
\tilde a_{j,w}^{(\ell,h)}
\right|
\le
L_a
\left(
j\|\Delta y_j^{(\ell-1)}\|_2
+
\sum_{w\le j}
\|\Delta y_w^{(\ell-1)}\|_2
\right).
\tag{32}
\]
The induction hypothesis gives
$j\|\Delta y_j^{(\ell-1)}\|_2
\le
C_{\ell-1}(\ln j)^{\ell-2}$. Moreover, by causality, (22), and the
induction hypothesis,
\[
\begin{aligned}
\sum_{w\le j}\|\Delta y_w^{(\ell-1)}\|_2
&\le
2M+
C_{\ell-1}
\sum_{w=i+1}^{j}
\frac{(\ln w)^{\ell-2}}{w} \\
&=
O\!\left((\ln j)^{\ell-1}\right).
\end{aligned}
\]
Since $j\ge2$, the first term in (32) can also be absorbed into the same
bound. Hence
\[
\sum_{w\le j}
\left|
a_{j,w}^{(\ell,h)}
-
\tilde a_{j,w}^{(\ell,h)}
\right|
=
O\!\left((\ln j)^{\ell-1}\right).
\]
Equation~(26) then implies
\[
\sum_{w\le j}
\left|
\alpha_{j,w}^{(\ell,h)}
-
\tilde\alpha_{j,w}^{(\ell,h)}
\right|
\le
C\frac{(\ln j)^{\ell-1}}{j}.
\]
Using the uniform bound (21),
\[
\begin{aligned}
\|II_{j,\ell,h}\|_2
&\le
\sum_{w\le j}
\left|
\alpha_{j,w}^{(\ell,h)}
-
\tilde\alpha_{j,w}^{(\ell,h)}
\right|
\|\tilde y_w^{(\ell-1)}\|_2
\|V_{\ell,h}\| \\
&\le
C\frac{(\ln j)^{\ell-1}}{j}.
\end{aligned}
\tag{33}
\]

Combining (31) and (33), we obtain
\[
\|b_{j,\ell,h}-\tilde b_{j,\ell,h}\|_2
\le
C\frac{(\ln j)^{\ell-1}}{j}.
\tag{34}
\]

Finally, by the Lipschitz continuity of $f^{(\ell)}$,
\[
\|\Delta y_j^{(\ell)}\|_2
\le
L_{\mathrm{act},\ell}
\left(
\|\Delta y_j^{(\ell-1)}\|_2
+
\sum_{h=1}^H
\|b_{j,\ell,h}-\tilde b_{j,\ell,h}\|_2
\right).
\]
Using (30) and (34),
\[
\|\Delta y_j^{(\ell)}\|_2
\le
L_{\mathrm{act},\ell}
\left(
C_{\ell-1}\frac{(\ln j)^{\ell-2}}{j}
+
HC\frac{(\ln j)^{\ell-1}}{j}
\right).
\]
Since $j\ge2$, both terms can be absorbed into a new constant $C_\ell>0$.
Therefore,
\[
\|\Delta y_j^{(\ell)}\|_2
\le
C_\ell\frac{(\ln j)^{\ell-1}}{j}.
\]
This completes the induction and proves (20).

\medskip
\noindent\textbf{Step 4: Output-layer bound.}

Taking $\ell=L$ and $j=n$ in (20), we obtain
\[
\|\Delta y_n^{(L)}\|_2
\le
C_L\frac{(\ln n)^{L-1}}{n}.
\]
By the definition of the output layer,
$\F(x)=y_n^{(L)}W_0+c$ and
$\F(x')=\tilde y_n^{(L)}W_0+c$. Hence
\[
\begin{aligned}
\|\F(x)-\F(x')\|_2
&=
\left\|
\bigl(y_n^{(L)}-\tilde y_n^{(L)}\bigr)W_0
\right\|_2 \\
&\le
\|W_0\|
\|\Delta y_n^{(L)}\|_2 \\
&\le
\|W_0\|C_L
\frac{(\ln n)^{L-1}}{n}.
\end{aligned}
\]
All constants appearing above depend only on the parameters and architecture
of the fixed transformer $\F$. Thus, setting
$C_{\F}:=\|W_0\|C_L$, we obtain, for every $x'\in B(x,1)$,
\[
\|\F(x)-\F(x')\|_2
\le
C_{\F}
\frac{(\ln n)^{L-1}}{n}.
\]
Taking the maximum over $x'\in B(x,1)$ gives
\[
\sen(\F,x)
=
\max_{x'\in B(x,1)}
\|\F(x)-\F(x')\|_2
\le
C_{\F}
\frac{(\ln n)^{L-1}}{n}.
\]
Since $n=\len(x)$,
\[
\sen(\F,x)
\le
C_{\F}
\frac{\ln^{L-1}(\len(x))}{\len(x)}.
\]
This completes the proof.
\end{proof}

\subsection{Proof of Theorem~\ref{th-2}}
\label{app:proof_thm5}

\begin{proof}[Proof of Theorem~\ref{th-2}]
We prove the theorem by explicitly constructing a Transformer whose
sensitivity matches the upper bound in Theorem~1 up to a constant factor.

To avoid confusion between the depth of the Transformer and the input
length, denote the prescribed depth by $D$. We construct a fixed
Transformer $F$ of depth $D$, with one attention head in each layer and
hidden dimension $2$. The feedforward network in every layer is set to
the zero map.

Let
\[
e=(1,-1)\in\mathbb{R}^2.
\]
Choose two distinct symbols $a,b\in\Sigma$, and define their embeddings by
\[
v_a=\eta e,
\qquad
v_b=0,
\]
where $\eta>0$ is a fixed constant.

In every attention layer, choose
\[
Q=0,
\qquad
K=0.
\]
Therefore, at every position $j$, all causal attention logits are equal.
The causal softmax weights are consequently uniform over all visible
positions:
\begin{equation}
\alpha_{j,w}
=
\frac{1}{j},
\qquad
1\leq w\leq j.
\label{eq:tight-uniform-attention}
\end{equation}

Choose the value matrix to be the identity:
\[
V=I_2.
\]
Since the feedforward network is zero, each layer consists only of the
causal attention operation together with the residual connection.

For every $n\geq2$, consider the pair of inputs
\[
x_n=(a,b,b,\ldots,b)\in\Sigma^n
\]
and
\[
x_n'=(b,b,b,\ldots,b)\in\Sigma^n.
\]
They differ only at the first position. Hence
\begin{equation}
x_n'\in B(x_n,1).
\label{eq:tight-input-pair}
\end{equation}

For the input $x_n'$, all initial embeddings are zero. Since the value
map is linear, the attention output of zero hidden states is again zero,
and the residual connection preserves zero. Thus every hidden state
remains zero at every layer. In particular,
\begin{equation}
F(x_n')=0.
\label{eq:tight-zero-input-output}
\end{equation}

We now consider $x_n$. Since the only nonzero initial embedding is a
multiple of $e$, and all operations in the constructed Transformer are
linear combinations of previous hidden states, every hidden state remains
in the one-dimensional subspace spanned by $e$.

Hence, for the hidden state at position $j$ after layer $\ell$, we may write
\[
y_j^{(\ell)}
=
m_j^{(\ell)}e,
\]
for some scalar $m_j^{(\ell)}$.

At layer $0$, the embeddings give
\begin{equation}
m_1^{(0)}=\eta,
\qquad
m_j^{(0)}=0
\quad
\text{for }j>1.
\label{eq:tight-initial-coefficients}
\end{equation}

By \eqref{eq:tight-uniform-attention}, the attention output at position
$j$ in layer $\ell$ is
\[
\frac{1}{j}
\sum_{w=1}^j
y_w^{(\ell-1)}
=
\left(
\frac{1}{j}
\sum_{w=1}^j
m_w^{(\ell-1)}
\right)e.
\]
Together with the residual connection, this gives the recurrence
\begin{equation}
m_j^{(\ell)}
=
m_j^{(\ell-1)}
+
\frac{1}{j}
\sum_{w=1}^j
m_w^{(\ell-1)},
\qquad
1\leq\ell\leq D.
\label{eq:tight-coefficient-recurrence}
\end{equation}

Since all initial coefficients are nonnegative, it follows immediately
from \eqref{eq:tight-coefficient-recurrence} that
\[
m_j^{(\ell)}\geq0
\]
for every $j$ and every $0\leq\ell\leq D$.

We next prove by induction on $\ell$ that, for every
$1\leq\ell\leq D$, there exist constants $c_\ell>0$ and
$N_\ell\in\mathbb{N}$ such that
\begin{equation}
m_j^{(\ell)}
\geq
c_\ell
\frac{(\ln j)^{\ell-1}}{j}
\qquad
\text{for every }j\geq N_\ell.
\label{eq:tight-induction-claim}
\end{equation}

For $\ell=1$, using \eqref{eq:tight-initial-coefficients} and
\eqref{eq:tight-coefficient-recurrence}, for every $j>1$ we have
\[
m_j^{(1)}
=
m_j^{(0)}
+
\frac{1}{j}
\sum_{w=1}^j
m_w^{(0)}
=
\frac{\eta}{j}.
\]
Therefore \eqref{eq:tight-induction-claim} holds for $\ell=1$ with,
for example,
\[
c_1=\eta.
\]

Now suppose that \eqref{eq:tight-induction-claim} holds for layer
$\ell-1$, where $2\leq\ell\leq D$. Since all coefficients are
nonnegative, \eqref{eq:tight-coefficient-recurrence} implies
\[
m_j^{(\ell)}
\geq
\frac{1}{j}
\sum_{w=1}^j
m_w^{(\ell-1)}.
\]
By the induction hypothesis, for all sufficiently large $j$,
\[
m_j^{(\ell)}
\geq
\frac{c_{\ell-1}}{j}
\sum_{w=N_{\ell-1}}^j
\frac{(\ln w)^{\ell-2}}{w}.
\]

Using the standard integral comparison,
\begin{align}
\sum_{w=N_{\ell-1}}^j
\frac{(\ln w)^{\ell-2}}{w}
&\geq
\int_{N_{\ell-1}}^j
\frac{(\ln t)^{\ell-2}}{t}\,dt
\nonumber\\
&=
\frac{
(\ln j)^{\ell-1}
-
(\ln N_{\ell-1})^{\ell-1}
}{
\ell-1
}.
\label{eq:tight-harmonic-integral}
\end{align}

Hence, for all sufficiently large $j$, there exists a constant
$C_\ell>0$, independent of $j$, such that
\begin{equation}
\sum_{w=N_{\ell-1}}^j
\frac{(\ln w)^{\ell-2}}{w}
\geq
C_\ell(\ln j)^{\ell-1}.
\label{eq:tight-harmonic-lower-bound}
\end{equation}

Consequently,
\[
m_j^{(\ell)}
\geq
c_{\ell-1}C_\ell
\frac{(\ln j)^{\ell-1}}{j}.
\]
Thus \eqref{eq:tight-induction-claim} holds at layer $\ell$ with
\[
c_\ell
=
c_{\ell-1}C_\ell>0.
\]
This completes the induction.

Taking $\ell=D$ and $j=n$ in
\eqref{eq:tight-induction-claim}, we obtain, for all sufficiently
large $n$,
\begin{equation}
m_n^{(D)}
\geq
c_D
\frac{(\ln n)^{D-1}}{n}.
\label{eq:tight-final-hidden-lower-bound}
\end{equation}

Finally, choose a linear output map that reads the coefficient in the
direction $e$. For example, let
\[
u=\frac{1}{2}(1,-1).
\]
Then
\[
\langle me,u\rangle=m
\]
for every $m\in\mathbb{R}$.

Therefore, at the final position,
\[
F(x_n)=m_n^{(D)},
\qquad
F(x_n')=0.
\]
By \eqref{eq:tight-final-hidden-lower-bound},
\begin{equation}
\left\|
F(x_n)-F(x_n')
\right\|_2
\geq
c_D
\frac{(\ln n)^{D-1}}{n}.
\label{eq:tight-output-difference}
\end{equation}

Since $x_n'\in B(x_n,1)$ by
\eqref{eq:tight-input-pair}, the definition of sensitivity gives
\begin{align}
\operatorname{sen}(F,x_n)
&\geq
\left\|
F(x_n)-F(x_n')
\right\|_2
\nonumber\\
&\geq
c_D
\frac{(\ln n)^{D-1}}{n}.
\label{eq:tight-sensitivity-lower-bound}
\end{align}

Now let $M\in\mathbb{N}_+$ be arbitrary. Choose $n\geq M$ sufficiently
large so that \eqref{eq:tight-sensitivity-lower-bound} holds, and set
$x=x_n$. Then
\[
\len(x)=n\geq M
\]
and
\[
\operatorname{sen}(F,x)
\geq
c_F
\frac{
(\ln\len(x))^{D-1}
}{
\len(x)
},
\]
where
\[
c_F:=c_D>0
\]
depends only on the constructed Transformer $F$.

Therefore, the upper bound in Theorem~1 is tight up to a multiplicative
constant.
\end{proof}

\section{Incompleteness of Infinite-Precision Transformers under Positive Confidence}
\label{app:proof_thm6}
Proof of Theorem \ref{th-confd}.
\begin{proof}
Suppose, toward a contradiction, that there exists a fixed transformer $F$ of depth $L$ that memorizes $T$ with confidence $c>0$.

Since $T$ is non-tame, there exist constants $C>0$ and
$0\leq \alpha<1$ such that $T$ has infinitely many
$(C,\alpha)$-sensitive inputs. Since the input alphabet $\Sigma_1$
is finite, there are only finitely many strings of bounded length.
Hence, we may choose a sequence of sensitive halting inputs
$(x_n)_{n\geq 1}$ such that
\[
    \ell_n := \operatorname{len}(x_n) \longrightarrow \infty.
\]
By the definition of $(C,\alpha)$-sensitivity, for every $n$ there
exists another halting input $x_n'$ satisfying
\[
    x_n'
    \in
    B\!\left(
        x_n,
        \left\lceil C\ell_n^\alpha \right\rceil
    \right)
\]
and
\[
    T(x_n)\neq T(x_n').
\]
Set
\[
    K_n:=\left\lceil C\ell_n^\alpha\right\rceil.
\]
Then, for all $n$,
\[
    K_n\leq (C+1)\ell_n^\alpha.
    \tag{1}
\]

Let
\[
    y_n := \widehat F(x_n)\oplus \sigma_0,
    \qquad
    y_n' := \widehat F(x_n')\oplus \sigma_0,
\]
where $\sigma_0$ is the end-of-output token. Since $F$ memorizes $T$,
\[
    \widehat F_{\mathrm{out}}(x_n)=T(x_n),
    \qquad
    \widehat F_{\mathrm{out}}(x_n')=T(x_n').
\]
Therefore,
\[
    y_n\neq y_n'.
\]

Let
\[
    j_n:=\min\{j:y_n[j]\neq y_n'[j]\}
\]
be the first position at which the two complete autoregressive
generations disagree, and let
\[
    z_n
    :=
    y_n[1:j_n-1]
    =
    y_n'[1:j_n-1]
\]
be their common prefix before this disagreement. Write
\[
    a_n:=y_n[j_n],
    \qquad
    b_n:=y_n'[j_n],
\]
so that $a_n\neq b_n$. At least one of $a_n$ and $b_n$ is not
$\sigma_0$. Exchanging the roles of $x_n$ and $x_n'$ if necessary,
we may assume
\[
    a_n\neq \sigma_0.
\]

Consider the two autoregressive states
\[
    s_n:=x_n\oplus z_n,
    \qquad
    s_n':=x_n'\oplus z_n.
\]
They have the same length
\[
    N_n:=\operatorname{len}(s_n)
        =\operatorname{len}(s_n')
        \geq \ell_n.
\]
Moreover, appending the same prefix $z_n$ creates no new
disagreement. Hence
\[
    s_n'\in B(s_n,K_n).
    \tag{2}
\]

By Corollary~2, with the constant $C_F$ depending only on the fixed
transformer $F$,
\[
    \|F(s_n)-F(s_n')\|_2
    \leq
    K_n C_F
    \frac{(\ln N_n)^{L-1}}{N_n}.
    \tag{3}
\]
For all sufficiently large $t$, the function
\[
    t\longmapsto \frac{(\ln t)^{L-1}}{t}
\]
is decreasing. Since $N_n\geq \ell_n$, equations~(1)--(3) imply,
for all sufficiently large $n$,
\[
\begin{aligned}
    \|F(s_n)-F(s_n')\|_2
    &\leq
    (C+1)C_F\ell_n^\alpha
    \frac{(\ln \ell_n)^{L-1}}{\ell_n} \\
    &=
    (C+1)C_F
    \frac{(\ln \ell_n)^{L-1}}
         {\ell_n^{\,1-\alpha}}.
\end{aligned}
\tag{4}
\]
Since $\alpha<1$,
\[
    \frac{(\ln \ell_n)^{L-1}}
         {\ell_n^{\,1-\alpha}}
    \longrightarrow 0.
\]
Therefore,
\[
    \|F(s_n)-F(s_n')\|_2
    \longrightarrow 0.
    \tag{5}
\]

We now use the positive-confidence assumption. By construction,
$a_n$ is the next token generated by $F$ at the state $s_n$.
Since $a_n\neq\sigma_0$ and $F$ memorizes $T$ with confidence $c$,
we have
\[
    F_{a_n}(s_n)
    -
    \max_{k\neq a_n}F_k(s_n)
    >c.
\]
In particular,
\[
    F_{a_n}(s_n)-F_{b_n}(s_n)>c.
    \tag{6}
\]

On the other hand, $b_n$ is the token selected by the autoregressive
generation at the state $s_n'$. Hence
\[
    F_{b_n}(s_n')\geq F_{a_n}(s_n'),
\]
or equivalently,
\[
    F_{b_n}(s_n')-F_{a_n}(s_n')\geq 0.
    \tag{7}
\]
Adding equations~(6) and~(7) gives
\[
    \bigl(F_{a_n}(s_n)-F_{a_n}(s_n')\bigr)
    +
    \bigl(F_{b_n}(s_n')-F_{b_n}(s_n)\bigr)
    >c.
    \tag{8}
\]
Consequently, at least one of the two terms on the left-hand side
of equation~(8) is larger than $c/2$. Therefore,
\[
    \|F(s_n)-F(s_n')\|_\infty>\frac{c}{2},
\]
and hence
\[
    \|F(s_n)-F(s_n')\|_2>\frac{c}{2}.
    \tag{9}
\]

Equation~(9) contradicts equation~(5) for all sufficiently large
$n$. Therefore, no fixed transformer $F$ can memorize a non-tame
Turing machine $T$ with any fixed confidence $c>0$.
\end{proof}

\section{Almost-Everywhere Incompleteness of Transformers}
In this section, we will prove  Theorem~\ref{ae} and Corollary \ref{Tarithae}.

\subsection{Parameter Space of Fixed-Architecture Transformers}

Let the tuple of all trainable parameters be denoted by $\theta$:

    \[
    \theta =
    \underbrace{\left(V_{\rm{emb}},\{Q_{l,h},K_{l,h},V_{l,h}\}_{l,h}\right)}_{\text{Embedding and attention}},
    \underbrace{\left(\{E_{l,1},E_{l,2},b_l\},W_{\rm{\rm{out}}},c\right)}_{\text{Feedforward Network}},\underbrace{\left(\{C_i\}_{i=1}^{\left\lceil\frac{W}{2}\right\rceil}\right)}_{\text{Position Embedding}}.
    \]
Let $T=|\Sigma|$ denote the vocabulary size. The total number of parameters is determined by
the depth $D$, width $W$, number of attention heads $H$, and vocabulary size
$|\Sigma|$, which we denote by $P(W,H,D,|\Sigma|)$.

For a transformer with a fixed architecture $(W,H,D,\Sigma)$, the total number
of real-valued parameters is
    \[
    \begin{aligned}
    P(W,H,D,T)
    =&
    \underbrace{TW}_{\text{Embedding Layer}}
    +
    \underbrace{3DHW^2}_{\text{attention Layers}}
    +
    \underbrace{D(2W^2+W)}_{\text{Feedforward Network}}
    \\
    &+
    \underbrace{W(T+1)+(T+1)}_{\text{Output Layer}}+
    \underbrace{
    \left\lceil\frac{W}{2}\right\rceil
    }_{\text{RoPE}} .
    \end{aligned}
    \]
Equivalently,
    \[
    P(W,H,D,T)=TW+DW(3HW+2W+1)+(W+1)(T+1)+\lceil\frac{W}{2}\rceil.
    \]
Therefore, the parameter space corresponding to a fixed architecture
$(W,H,D,\Sigma)$ is defined as
    \[
    \Theta_{W,H,D,\Sigma}
    :=
    \left\{\theta:\theta\text{ is a parameter set of a transformer with architecture }(W,H,D,\Sigma)\right\}.
    \]
Without imposing any additional constraints on parameter magnitudes, this
parameter space is naturally isomorphic to a finite-dimensional Euclidean
space $\mathbb{R}^{P(W,H,D,T)}$

\subsection{Proof of Theorem~\ref{ae}}
\label{app:proof_thm7}

Before proceeding to the formal proof of the theorem, we first need to establish that the parameter set of interest is Borel measurable.

\begin{lemma}[Measurability of the memorization parameter set]
\label{lem:MG-borel}

Fix a Turing machine $\TM$, a transformer architecture with depth $D$,
width $W$, number of heads $H$, and finite vocabulary $\Sigma$. Let
$\Theta_{D,W,H,\Sigma}\cong\mathbb{R}^{P}$ be its parameter space.
Let $G:\Sigma^*\to 2^{\Sigma^*}$ be an $f$-bounded CoT rule, i.e.,
$|G(x)|\leq f(\len(x))$ for every halting input $x$ of $\TM$.

Assume that generation is performed using a fixed deterministic
$\arg\max$ tie-breaking rule. Then
\[
M_G(\TM)\subseteq\Theta_{D,W,H,\Sigma}
\]
is a Borel measurable set.
\end{lemma}

\begin{proof}
Fix a finite autoregressive state $s$ and an output token $a$, and let
$L_{s,a}(\theta):=F_\theta(s)_a$ denote the output logit corresponding to
token $a$. For fixed $s$ and $a$, the map
$\theta\mapsto L_{s,a}(\theta)$ is continuous.

Fix a deterministic tie-breaking rule that chooses the token with the
smallest index among all tokens attaining the maximal logit. Then the set
of parameters for which the next generated token at state $s$ is $a$ is
\[
\mathcal E_{s,a}
:=
\bigcap_{j<a}
\left\{
\theta:
L_{s,a}(\theta)-L_{s,j}(\theta)>0
\right\}
\cap
\bigcap_{j>a}
\left\{
\theta:
L_{s,a}(\theta)-L_{s,j}(\theta)\geq0
\right\}.
\]
Hence $\mathcal E_{s,a}$ is Borel.

For a halting input $x$, every admissible CoT $r\in G(x)$ determines an
admissible complete generation sequence
$y=r\oplus\sigma_T\oplus\TM(x)\oplus\sigma_0$, where $\sigma_T$ is the
CoT delimiter and $\sigma_0$ is the end-of-output symbol. Define
\[
\Gamma_G^\TM(x)
:=
\left\{
r\oplus\sigma_T\oplus\TM(x)\oplus\sigma_0:
r\in G(x)
\right\}.
\]
Then $|\Gamma_G^\TM(x)|=|G(x)|\leq f(\len(x))$, so
$\Gamma_G^\TM(x)$ is finite.

For a fixed complete sequence
$y=(y_1,\ldots,y_m)\in\Gamma_G^\TM(x)$, write
$y_{<j}:=(y_1,\ldots,y_{j-1})$ and define
\[
\mathcal E_{x,y}
:=
\bigcap_{j=1}^{m}
\mathcal E_{x\oplus y_{<j},\,y_j}.
\]
Since $m<\infty$, the set $\mathcal E_{x,y}$ is Borel.

The set of parameters that generate an admissible complete output on
input $x$ is therefore
\[
\mathcal E_x^G
:=
\bigcup_{y\in\Gamma_G^\TM(x)}
\mathcal E_{x,y}.
\]
Since $\Gamma_G^\TM(x)$ is finite, $\mathcal E_x^G$ is a finite union of
Borel sets and hence is Borel.

Finally,
\[
M_G(\TM)
=
\bigcap_{x\in\mathcal H_\TM}
\mathcal E_x^G,
\]
where $\mathcal H_\TM$ denotes the set of halting inputs of $\TM$.
Since $\mathcal H_\TM\subseteq\Sigma_1^*$ is countable, this is a
countable intersection of Borel sets. Therefore $M_G(\TM)$ is Borel
measurable.
\end{proof}

We now prove Theorem~\ref{ae}.

\begin{theorem}

  For arbitrary $S,D,W,H,\Sigma$, assume that the following conditions are satisfied.
    
(1) $\TM=(Q,\Sigma_1,\Sigma_2,\delta)$ is a $(C,\alpha)$ non-tame Turing machine and $\Sigma_2\subset \Sigma_0$.

(2) $\Alg(S,D,W,H,\Sigma)$ is an absolutely continuous distribution.

(3) $G$ is a $f$-number of rules, and $f$ satisfies $\lim_{n\to\infty}\frac{f(n)(\ln n)^{(D-1)/2}}{\sqrt{n^{1-\alpha}}}=0$.

Then, $\Pr_{\F\sim \Alg(S,D,W,H,\Sigma)}(\F\in M_G(\TM))=0$.
That is, if fewer than $\widetilde{o}(\sqrt{n^{1-\alpha}})$ CoTs are allowed for inputs of length $n$, then the probability for $\F$ trained by algorithm $\Alg$ to memorize $\TM$ is zero.
\end{theorem}

\begin{proof}
Let $\Theta_{D,W,H,\Sigma}\cong\mathbb{R}^{P}$ be the parameter space of
the fixed transformer architecture, equipped with the $P$-dimensional
Lebesgue measure $\lambda_P$. By Lemma~\ref{lem:MG-borel},
$M_G(\TM)$ is Borel measurable.

We first prove $\lambda_P(M_G(\TM))=0$, and then use the absolute
continuity of the randomized training distribution.

\medskip
\noindent\textbf{Step 1: Choose a sequence of sensitive input pairs.}

Since $\TM$ is non-tame with coefficients $C>0$ and $0\leq\alpha<1$,
it has infinitely many $(C,\alpha)$-sensitive halting inputs. Since the
input alphabet $\Sigma_1$ is finite, there are only finitely many strings
whose lengths are bounded by any fixed constant. Therefore, we may choose
a sequence of sensitive inputs $(x_n)_{n\geq1}$ such that
\begin{equation}
\ell_n:=\len(x_n)\longrightarrow\infty.
\label{eq:sensitive-input-length-diverges}
\end{equation}

By the definition of $(C,\alpha)$-sensitivity, for every $n$ there exists
another halting input
$x_n'\in B(x_n,\lceil C\ell_n^\alpha\rceil)$ such that
\begin{equation}
\TM(x_n)\neq\TM(x_n').
\label{eq:sensitive-output-difference}
\end{equation}
By the definition of $B(x,k)$,
\begin{equation}
\len(x_n')=\len(x_n)=\ell_n,
\label{eq:sensitive-pair-equal-length}
\end{equation}
and
\begin{equation}
d_{\mathrm{nf}}(x_n,x_n')
\leq
(C+1)\ell_n^\alpha.
\label{eq:sensitive-pair-distance}
\end{equation}

\medskip
\noindent\textbf{Step 2: Fix all parameters except the output bias.}

Let $K$ denote the size of the output vocabulary. Separate the output bias
from all remaining transformer parameters and write
$\theta=(\vartheta,c)\in\mathbb{R}^{P-K}\times\mathbb{R}^{K}$. Then,
for every finite autoregressive state $s$,
$F_{\vartheta,c}(s)=G_\vartheta(s)+c$, where $G_\vartheta(s)$ does not
depend on $c$.

Fix $\vartheta\in\mathbb{R}^{P-K}$ and define the corresponding
output-bias section
\[
M_{G,\vartheta}(\TM)
:=
\left\{
c\in\mathbb{R}^{K}:
(\vartheta,c)\in M_G(\TM)
\right\}.
\]
We shall prove that $\lambda_K(M_{G,\vartheta}(\TM))=0$.

\medskip
\noindent\textbf{Step 3: Every pair of admissible CoTs produces a thin slab.}

Fix $n$ and take arbitrary complete admissible outputs
$y\in\Gamma_G^\TM(x_n)$ and $y'\in\Gamma_G^\TM(x_n')$. By
\eqref{eq:sensitive-output-difference}, the correct Turing-machine outputs
on $x_n$ and $x_n'$ are different. Since every sequence in
$\Gamma_G^\TM(x)$ contains the corresponding correct output, we have
$y\neq y'$.

Let $j(y,y'):=\min\{j:y[j]\neq y'[j]\}$ be their first disagreement
position, and let
$z(y,y'):=y[1:j(y,y')-1]=y'[1:j(y,y')-1]$ be the common prefix preceding
their first disagreement. Write
$a(y,y'):=y[j(y,y')]$ and $b(y,y'):=y'[j(y,y')]$, so that
$a(y,y')\neq b(y,y')$.

Define
\[
s_{n,y,y'}
:=
x_n\oplus z(y,y'),
\qquad
s_{n,y,y'}'
:=
x_n'\oplus z(y,y'),
\]
and let $N_{n,y,y'}:=\len(s_{n,y,y'})$. Then
\begin{equation}
N_{n,y,y'}
=
\len(s_{n,y,y'}')
\geq
\ell_n.
\label{eq:autoregressive-state-length}
\end{equation}

Appending the same prefix to $x_n$ and $x_n'$ does not introduce any new
disagreement. Hence, by \eqref{eq:sensitive-pair-distance},
\[
d_{\mathrm{nf}}
\left(
s_{n,y,y'},
s_{n,y,y'}'
\right)
\leq
(C+1)\ell_n^\alpha.
\]
Let
$k_n:=d_{\mathrm{nf}}(x_n,x_n')\leq(C+1)\ell_n^\alpha$. We may construct
a sequence
$s^{(0)}=s_{n,y,y'},s^{(1)},\ldots,s^{(k_n)}=s_{n,y,y'}'$ such that every
consecutive pair differs at exactly one non-final position.

Applying the single-token sensitivity bound to the transformer with
parameters $(\vartheta,0)$, and then using the triangle inequality, gives
\[
\left\|
G_\vartheta(s_{n,y,y'})
-
G_\vartheta(s_{n,y,y'}')
\right\|_2
\leq
C_{(\vartheta,0)}
k_n
\frac{(\ln N_{n,y,y'})^{D-1}}
{N_{n,y,y'}}.
\]
Since $k_n\leq(C+1)\ell_n^\alpha$ and the function
$t\mapsto(\ln t)^{D-1}/t$ is decreasing for all sufficiently large $t$,
it follows from \eqref{eq:autoregressive-state-length} that, for all
sufficiently large $n$,
\begin{equation}
\left\|
G_\vartheta(s_{n,y,y'})
-
G_\vartheta(s_{n,y,y'}')
\right\|_2
\leq
(C+1)C_{(\vartheta,0)}
\frac{(\ln\ell_n)^{D-1}}
{\ell_n^{1-\alpha}}.
\label{eq:uniform-sensitive-logit-bound}
\end{equation}

Now suppose that an output bias $c$ makes the transformer generate $y$ on
$x_n$ and $y'$ on $x_n'$. For simplicity, write
$a:=a(y,y')$ and $b:=b(y,y')$.

At state $s_{n,y,y'}$, the next generated token must be $a$, so
\[
c_a-c_b
\geq
G_\vartheta(s_{n,y,y'})_b
-
G_\vartheta(s_{n,y,y'})_a.
\]
Similarly, at state $s_{n,y,y'}'$, the next generated token must be $b$, so
\[
c_a-c_b
\leq
G_\vartheta(s_{n,y,y'}')_b
-
G_\vartheta(s_{n,y,y'}')_a.
\]
Thus $c_a-c_b\in I_{n,y,y'}(\vartheta)$, where
\[
I_{n,y,y'}(\vartheta)
:=
\left[
G_\vartheta(s_{n,y,y'})_b
-
G_\vartheta(s_{n,y,y'})_a,
\,
G_\vartheta(s_{n,y,y'}')_b
-
G_\vartheta(s_{n,y,y'}')_a
\right].
\]
If the left endpoint is larger than the right endpoint, then no output bias
$c$ can generate this pair of CoTs, and the corresponding feasible set is
empty.

Whenever $I_{n,y,y'}(\vartheta)$ is nonempty,
\begin{align}
|I_{n,y,y'}(\vartheta)|
&\leq
\left|
G_\vartheta(s_{n,y,y'}')_b
-
G_\vartheta(s_{n,y,y'})_b
\right|
\nonumber\\
&\quad+
\left|
G_\vartheta(s_{n,y,y'}')_a
-
G_\vartheta(s_{n,y,y'})_a
\right|
\nonumber\\
&\leq
2
\left\|
G_\vartheta(s_{n,y,y'})
-
G_\vartheta(s_{n,y,y'}')
\right\|_2.
\label{eq:interval-width-by-logit-distance}
\end{align}
Combining \eqref{eq:uniform-sensitive-logit-bound} and
\eqref{eq:interval-width-by-logit-distance}, we obtain
\begin{equation}
|I_{n,y,y'}(\vartheta)|
\leq
2(C+1)C_{(\vartheta,0)}
\frac{(\ln\ell_n)^{D-1}}
{\ell_n^{1-\alpha}}
\label{eq:uniform-interval-width-tm}
\end{equation}
for all sufficiently large $n$, uniformly over all admissible pairs
$(y,y')$.

\medskip
\noindent\textbf{Step 4: Cover the feasible section by finitely many thin slabs.}

For distinct tokens $a,b$ and an interval $I\subseteq\mathbb{R}$, define
$\mathcal H_{a,b}(I):=\{c\in\mathbb{R}^{K}:c_a-c_b\in I\}$.

Let $\mathcal S_{n,\vartheta}$ be the set of output biases for which the
transformer generates some admissible complete output on both $x_n$ and
$x_n'$. By the first-disagreement argument above,
\begin{equation}
\mathcal S_{n,\vartheta}
\subseteq
\bigcup_{\substack{
y\in\Gamma_G^\TM(x_n)\\
y'\in\Gamma_G^\TM(x_n')
}}
\mathcal H_{a(y,y'),b(y,y')}
\left(
I_{n,y,y'}(\vartheta)
\right).
\label{eq:slab-cover-tm}
\end{equation}

Fix $R>0$ and let $Q_R:=[-R,R]^K$. For every $a\neq b$ and every finite
interval $I$, Fubini's theorem gives
\begin{equation}
\lambda_K
\left(
\mathcal H_{a,b}(I)\cap Q_R
\right)
\leq
(2R)^{K-1}|I|.
\label{eq:single-slab-measure-tm}
\end{equation}

By the $f$-boundedness of $G$ and
\eqref{eq:sensitive-pair-equal-length},
\begin{equation}
|\Gamma_G^\TM(x_n)|
\,
|\Gamma_G^\TM(x_n')|
\leq
f(\ell_n)^2.
\label{eq:number-of-cot-pairs-tm}
\end{equation}

Combining \eqref{eq:slab-cover-tm},
\eqref{eq:single-slab-measure-tm},
\eqref{eq:uniform-interval-width-tm}, and
\eqref{eq:number-of-cot-pairs-tm}, we obtain, for all sufficiently large
$n$,
\begin{align}
\lambda_K
\left(
\mathcal S_{n,\vartheta}\cap Q_R
\right)
&\leq
2(C+1)C_{(\vartheta,0)}
(2R)^{K-1}
f(\ell_n)^2
\frac{(\ln\ell_n)^{D-1}}
{\ell_n^{1-\alpha}}
\nonumber\\
&=
2(C+1)C_{(\vartheta,0)}
(2R)^{K-1}
\left[
\frac{
f(\ell_n)(\ln\ell_n)^{(D-1)/2}
}{
\sqrt{\ell_n^{1-\alpha}}
}
\right]^2.
\label{eq:section-cover-measure-bound-tm}
\end{align}

By assumption,
\[
\lim_{n\to\infty}
\frac{
f(n)(\ln n)^{(D-1)/2}
}{
\sqrt{n^{1-\alpha}}
}
=0.
\]
Together with \eqref{eq:sensitive-input-length-diverges}, this implies from
\eqref{eq:section-cover-measure-bound-tm} that
\begin{equation}
\lambda_K
\left(
\mathcal S_{n,\vartheta}\cap Q_R
\right)
\longrightarrow0.
\label{eq:Sn-measure-goes-zero-tm}
\end{equation}

Every $c\in M_{G,\vartheta}(\TM)$ must correctly handle both $x_n$ and
$x_n'$ for every $n$. Hence
$M_{G,\vartheta}(\TM)\subseteq\mathcal S_{n,\vartheta}$ for every $n$.
Therefore,
\[
\lambda_K
\left(
M_{G,\vartheta}(\TM)\cap Q_R
\right)
\leq
\lambda_K
\left(
\mathcal S_{n,\vartheta}\cap Q_R
\right)
\]
for every $n$. Letting $n\to\infty$ and using
\eqref{eq:Sn-measure-goes-zero-tm}, we obtain
\begin{equation}
\lambda_K
\left(
M_{G,\vartheta}(\TM)\cap Q_R
\right)
=0.
\label{eq:bounded-section-zero-tm}
\end{equation}

Since $\mathbb{R}^{K}=\bigcup_{m=1}^{\infty}[-m,m]^K$, it follows from
\eqref{eq:bounded-section-zero-tm} that
\begin{equation}
\lambda_K
\left(
M_{G,\vartheta}(\TM)
\right)
=0
\qquad
\text{for every }
\vartheta\in\mathbb{R}^{P-K}.
\label{eq:all-output-bias-sections-zero}
\end{equation}

\medskip
\noindent\textbf{Step 5: The full memorization parameter set has measure zero.}

By Lemma~\ref{lem:MG-borel}, $M_G(\TM)$ is Borel measurable. Hence
Tonelli's theorem and \eqref{eq:all-output-bias-sections-zero} imply
\begin{align}
\lambda_P(M_G(\TM))
=
\int_{\mathbb{R}^{P-K}}
\lambda_K
\left(
M_{G,\vartheta}(\TM)
\right)
\,d\vartheta
\nonumber = 0 
\label{eq:MG-lebesgue-zero}
\end{align}

\medskip
\noindent\textbf{Step 6: Apply absolute continuity of the training distribution.}

Let $\mu:=\Alg(S,D,W,H,\Sigma)$ be the probability distribution on
transformer parameters produced by the randomized training algorithm. By
assumption, $\mu\ll\lambda_P$. Since$\lambda_P(M_G(\TM))=0$, absolute
continuity implies $\mu(M_G(\TM))=0$. Therefore,
\[
\Pr_{\F\sim\Alg(S,D,W,H,\Sigma)}
\left(
\F\in M_G(\TM)
\right)
=0.
\]

This completes the proof.
\end{proof}

\subsection{Proof of Corollary \ref{Tarithae}}
\label{app-Tarithae}
We now restate Corollary \ref{Tarithae} and provide a proof.

\begin{corollary}
For $\TM_\Ar$, we restrict the generation rules $G$ to performing one calculation at a time in CoT. That is, at each CoT step, the leftmost currently calculable innermost operation is evaluated and replaced by its value in a fixed format (shown in Example~\ref{eg:mem_turing}). Then, if $\Alg(S,D,W,H,\Sigma)$ is an absolutely continuous distribution, then $\Pr_{\F\sim \Alg(S,D,W,H,\Sigma)}(\F\in M_G(\TM_\Ar))=0$.
\end{corollary}

\begin{proof}
First, note that  $\TM_\Ar$ is non-tame since it contains infinitely many $(1,0)$-sensitive inputs.

Let $\Theta_{D,W,H,\Sigma}\cong\mathbb{R}^{P}$ be the parameter space of
the fixed Transformer architecture, and let $K$ be the size of the output
vocabulary. We first show that $\lambda_P(M_G(\TM_\Ar))=0$.

We construct a sequence of arithmetic input pairs for which the admissible
CoT is unique. Let $A_1=0$ and $A_{n+1}=(A_n+0)$, and define
$x_n=(A_n+1)$ and $x_n'=(A_n*1)$. Then $x_n$ and $x_n'$ have the same
length, differ at exactly one non-final token, and satisfy
$\TM_\Ar(x_n)=1$ and $\TM_\Ar(x_n')=0$. Moreover,
$\len(x_n)\to\infty$.

Under the generation rule $G$, each CoT step computes one currently
calculable operator. By the nested structure above, at every step, there is
exactly one such operator. Hence, both $x_n$ and $x_n'$ admit a unique
admissible CoT. Let $y_n$ and $y_n'$ denote the corresponding complete
generation sequences. Since their final outputs differ, $y_n\neq y_n'$.

As in the proof above, separate the output bias and write
$\theta=(\vartheta,c)\in\mathbb{R}^{P-K}\times\mathbb{R}^{K}$, so that
$F_{\vartheta,c}(s)=G_\vartheta(s)+c$. Fix
$\vartheta\in\mathbb{R}^{P-K}$ and let
\[
M_{G,\vartheta}(\TM_\Ar)
:=
\left\{
c\in\mathbb{R}^{K}:
(\vartheta,c)\in M_G(\TM_\Ar)
\right\}.
\]

For each $n$, let $j_n$ be the first position at which $y_n$ and $y_n'$
differ, let $z_n$ be their common prefix before $j_n$, and write
$a_n:=y_n[j_n]$ and $b_n:=y_n'[j_n]$. Since the output vocabulary is
finite, after passing to an infinite subsequence, we may assume that
$a_n=a$ and $b_n=b$ for two fixed tokens $a\neq b$.

Set $s_n:=x_n\oplus z_n$ and $s_n':=x_n'\oplus z_n$. These two states
still differ at exactly one non-final position and have the same length
$N_n$, where $N_n\geq\len(x_n)\to\infty$. By Theorem~4 applied to the
Transformer with parameters $(\vartheta,0)$,
\[
\left\|
G_\vartheta(s_n)-G_\vartheta(s_n')
\right\|_2
\leq
C_\vartheta
\frac{(\ln N_n)^{D-1}}{N_n}
\longrightarrow 0.
\]

Now let $c\in M_{G,\vartheta}(\TM_\Ar)$. Since the admissible CoTs are
unique, the Transformer must generate $y_n$ on $x_n$ and $y_n'$ on
$x_n'$. At their first disagreement, it must therefore choose $a$ at
$s_n$ and $b$ at $s_n'$. Hence
\[
c_a-c_b\in I_n(\vartheta),
\]
where
\[
I_n(\vartheta)
:=
\left[
G_\vartheta(s_n)_b-G_\vartheta(s_n)_a,\,
G_\vartheta(s_n')_b-G_\vartheta(s_n')_a
\right].
\]
As in the preceding proof, whenever $I_n(\vartheta)$ is nonempty,
\[
|I_n(\vartheta)|
\leq
2
\left\|
G_\vartheta(s_n)-G_\vartheta(s_n')
\right\|_2
\longrightarrow 0.
\]

Thus, for any $c,\widetilde c\in M_{G,\vartheta}(\TM_\Ar)$,
both $c_a-c_b$ and $\widetilde c_a-\widetilde c_b$ belong to
$I_n(\vartheta)$ for every $n$ in the chosen subsequence. Therefore
\[
\left|
(c_a-c_b)-(\widetilde c_a-\widetilde c_b)
\right|
\leq
|I_n(\vartheta)|
\longrightarrow 0,
\]
and hence
$c_a-c_b=\widetilde c_a-\widetilde c_b$.

Consequently, either $M_{G,\vartheta}(\TM_\Ar)=\varnothing$, or there
exists $r_\vartheta\in\mathbb{R}$ such that
\[
M_{G,\vartheta}(\TM_\Ar)
\subseteq
\left\{
c\in\mathbb{R}^{K}:
c_a-c_b=r_\vartheta
\right\}.
\]
The set on the right is an affine hyperplane in $\mathbb{R}^{K}$ and hence
has $K$-dimensional Lebesgue measure zero. Thus
\[
\lambda_K\!\left(M_{G,\vartheta}(\TM_\Ar)\right)=0
\qquad
\text{for every }
\vartheta\in\mathbb{R}^{P-K}.
\]

By Lemma~\ref{lem:MG-borel}, $M_G(\TM_\Ar)$ is Borel measurable. Hence,
by Tonelli's theorem,
\[
\lambda_P\!\left(M_G(\TM_\Ar)\right)
=
\int_{\mathbb{R}^{P-K}}
\lambda_K\!\left(M_{G,\vartheta}(\TM_\Ar)\right)
\,d\vartheta
=
0.
\]

Finally, since $\Alg(S,D,W,H,\Sigma)\ll\lambda_P$, absolute continuity gives
\[
\Pr_{F\sim \Alg(S,D,W,H,\Sigma)}
\left(
F\in M_G(\TM_\Ar)
\right)
=0.
\]
This completes the proof.
\end{proof}

\section{Turing Completeness of Agents}

In this section, we will prove Theorem~\ref{agent} and a version of Theorem~\ref{agent} for agents with one transformer.

\subsection{Proof of Theorem~\ref{agent}}
\label{app:proof_thm8}

To prove the theorem, we first introduce several auxiliary functions.

\begin{definition}
\textbf{Search function.}
For $x,z\in\Sigma^*$ with $|z|\geq 3$, let
\[
i
=
\operatorname*{argmin}_{r}\{r:x[r]=z[2]\},
\qquad
j
=
\operatorname*{argmin}_{r}\{r:x[r]=z[3]\}.
\]
We use the convention that $x[r]=B$ whenever $r\leq 0$ or
$r\geq |x|+1$. Define
\[
f_{\mathrm{Search}}(x,z)
=
x[i-1:j+1].
\]
If either $i$ or $j$ does not exist, set $f_{\mathrm{Search}}(x,z)=B \oplus x[1]$.

\textbf{Judge function.}
For a fixed symbol $\sigma\in\Sigma$, define
\[
f_{\mathrm{Judge},\sigma}(x)
=
\begin{cases}
1, & x[|x|]=\sigma,\\
0, & x[|x|]\neq\sigma.
\end{cases}
\]

\textbf{Delete function.}
For $x,z\in\Sigma^*$, define
\[
f_{\mathrm{Delete}}(x,z)=x/z,
\]
where $x/z$ denotes the sequence obtained from $x$ by deleting every
symbol $x[r]$ that coincides with at least one symbol occurring in $z$.

\textbf{Replace function.}
For $x,y,z\in\Sigma^*$ with $|y|\ge 3$ and $|z|\ge 3$, let
\[
i=\arg\min_r\{r:x[r]=y[2]\},
\qquad
j=\arg\min_r\{r:x[r]=y[3]\}.
\]
If either index does not exist, define
\[
f_{\mathrm{Replace}}(x,y,z)=x.
\]
Otherwise, define
\[
f_{\mathrm{Replace}}(x,y,z)
=
x[1:i-2]
\oplus
z[2:|z|-1]
\oplus
x[j+2:|x|].
\]

\end{definition}

Each of the four functions above can be computed in linear time.

\begin{definition}
\label{df9}
Let
\[
\mathcal{T}
=
(Q,\Sigma_1,\Sigma_2,\delta)
\]
be a Turing machine. For an input $x$, let $\mathcal{T}(x)[i]$ denote
the finite tape segment recorded after the $i$-th transition as follows.
Let $l_1$ and $r_1$ be the leftmost and rightmost non-blank tape positions,
and let $l_2$ and $r_2$ be the leftmost and rightmost positions ever
visited by the head up to that time. Then $\mathcal{T}(x)[i]$ is the tape
segment between
\[
\min\{l_1,l_2\}
\qquad\text{and}\qquad
\max\{r_1,r_2\}.
\]
In particular,
\[
\mathcal{T}(x)[0]=x.
\]
\end{definition}

\begin{definition}
For $z\in\mathbb{Z}_+$, define
\[
B_z:
\mathbb{Z}_+
\longrightarrow
\{-1,0,1\}^{\lceil\ln_2 z\rceil}
\]
as follows. For $x\in\{1,\ldots,z-1\}$, $B_z(x)$ is the binary
representation of $x$ with length $\lceil\ln_2 z\rceil$. Moreover,
\[
B_z(z)
=
-\mathbf{1}_{\lceil\ln_2 z\rceil},
\]
and for all remaining values of $x$ we set
\[
B_z(x)
=
\mathbf{0}_{\lceil\ln_2 z\rceil}.
\]
\end{definition}

\begin{definition}
Assume that
\[
\Sigma_2
\subset
\Sigma\setminus\{\sigma_b,\sigma_e,0,1,-1\}.
\]
Let $\mathcal{T}(x)[i]$ be as in Definition~\ref{df9}. Suppose that,
after the $i$-th transition, the machine is in state $q_k$ and its head
is at position $j$ of $\mathcal{T}(x)[i]$. Then define
\begin{align}
\mathcal{T}'(x)[i]
&=
\mathcal{T}(x)[i][1:j-1]
\oplus
\sigma_b
\oplus
B_{|Q|}(k)
\oplus
\sigma_e
\nonumber\\
&\quad\oplus
\mathcal{T}(x)[i][j:|\mathcal{T}(x)[i]|].
\label{eq:encoded-configuration}
\end{align}
If $j=0$, we use the convention
\[
\mathcal{T}(x)[i][-1]=B.
\]
Without loss of generality, let
\[
\sigma_{|\Sigma_2|}=B.
\]
Thus $\mathcal{T}'(x)[i]$ explicitly records both the head position and
the current state.
\end{definition}

We also use the following notation.

\begin{definition}
For indices $i_1,\ldots,i_r$, let
\[
I_{i_1,\ldots,i_r}
\]
denote the vector whose coordinates $i_1,\ldots,i_r$ are equal to $1$
and whose remaining coordinates are $0$.

Similarly,
\[
I_{(i_1,j_1),\ldots,(i_r,j_r)}
\]
denotes the matrix whose entries at
$(i_1,j_1),\ldots,(i_r,j_r)$ are equal to $1$ and whose remaining
entries are $0$.
\end{definition}

We now establish two auxiliary lemmas.

\begin{lemma}
\label{lm1}
Let $\Sigma=\{\sigma_i\}_{i=1}^{T}$. For any
$k,j\in\{1,\ldots,T\}$ and any $q\geq6$, there exists a Transformer
$\mathcal{F}$ of width $3$, depth $2$, and one attention head such that
the following holds for every $x\in\Sigma^*$.

If $1\leq\#\{r:x[r]=\sigma_k\}\leq\lfloor e^q\rfloor$, then
$\mathcal{F}(x)=I_j$ . If $\sigma_k$ does not
occur in $x$, then $\mathcal{F}(x)=0$.
\end{lemma}

\begin{proof}
We construct the Transformer explicitly. For the embedding layer, embed
each symbol $\sigma_i$ as
\[
v_i=
\begin{cases}
(0,1,q), & i=k,\\
(0,0,q), & i\neq k.
\end{cases}
\]

In the first hidden layer, choose $A_{m-n}=I$,
$QK=qI_{(3,2)}$, $V=I_{(2,1)}$, and $\operatorname{FNN}=0$. Then
\[
v_iQKv_r^{\top}
=
\begin{cases}
q^2, & \sigma_r=\sigma_k,\\
0, & \sigma_r\neq\sigma_k.
\end{cases}
\]

Suppose that $\sigma_k$ appears exactly $m$ times, where
$1\leq m\leq\lfloor e^q\rfloor$. After subtracting the maximum attention
score, the exponentiated score is $1$ at positions containing $\sigma_k$
and $e^{-q^2}$ elsewhere. Since $q\geq6$, $[e^{-q^2}]_q=0$. Hence only
the $m$ occurrences of $\sigma_k$ receive nonzero attention weight.

Let $\beta_x:=m[1/m]_q$. Since
$\left|[1/m]_q-1/m\right|\leq\frac12 10^{-q}$, we obtain
\[
|\beta_x-1|
\leq
\frac{m}{2}10^{-q}
\leq
\frac12\left(\frac{e}{10}\right)^q
\leq
\frac14.
\]
Therefore $\beta_x\in[\frac34,\frac54]$.

Hence, the first coordinate of the hidden state at the final position
satisfies
\[
h_1(x)[1]
=
\begin{cases}
\beta_x, & \sigma_k\text{ occurs in }x,\\
0, & \text{otherwise}.
\end{cases}
\]

Finally, choose the output map
$\phi(t)=2\operatorname{ReLU}(t-\tfrac14)
-2\operatorname{ReLU}(t-\tfrac34)$.
Then $\phi(0)=0$ and $\phi(\beta)=1$ for
$\beta\in[\frac34,\frac54]$. Choosing all output coordinates except $j$
to be zero therefore gives $\mathcal{F}(x)[j]=1$. Hence
\[
\mathcal{F}(x)
=
\begin{cases}
I_j,
& 1\leq\#\{r:x[r]=\sigma_k\}\leq\lfloor e^q\rfloor,\\
0,
& \sigma_k\text{ does not occur}.
\end{cases}
\]
This proves the lemma.
\end{proof}

Next, an important lemma will be presented for mapping a symbol back to its index within the alphabet.

\begin{lemma}
\label{lm2}
Let $L\ge2$, let $|\Sigma|$ be the vocabulary size, and set
$P=10^{\lceil\ln_{10}(16L)\rceil}$, $\omega=P^{-1}$, and
$A=64(q+1)P^2$. Suppose $q\ge6$ and
\begin{equation}
10^q\ge 10^4(q+1)P^2.
\label{eq:copy-precision}
\end{equation}
There is one attention layer of width $O(L|\Sigma|)$ with $L$ heads such
that, for every $1\le n=|x|\le L$, its final-position hidden state contains
blocks $V_0,\ldots,V_{L-1}$ satisfying
\begin{equation}
V_h=e_{x[n-\min\{h,n-1\}]}.
\label{eq:copy-exact}
\end{equation}
In particular, $V_h=e_{x[n-h]}$ for every $0\le h\le n-1$. For inputs of
arbitrary length, the same layer also satisfies
\begin{equation}
(V_h)_i\ge0,
\qquad
\norm{V_h}_1\le2,
\qquad
0\le h\le n-1.
\label{eq:copy-safe}
\end{equation}
\end{lemma}

\begin{proof}
For the embedding layer, embed each token $a\in\Sigma$ as
$v_a=(1,e_a,0,\ldots,0)$, where $e_a\in\mathbb{R}^K$ is its one-hot
representation. For each attention head, reserve a disjoint block of $K$
coordinates in the hidden state, initialized to zero.

For the positional encoding, choose a two-dimensional RoPE plane with
frequency $\omega=P^{-1}$ and set all other required frequencies to zero.
Define $c_u=\cos_q(u\omega)$ and $s_u=\sin_q(u\omega)$ for $u\ge0$.

The attention layer contains $L$ heads. For the $h$-th head,
$0\le h\le L-1$, choose the query corresponding to the constant coordinate
as $(Ac_h,-As_h)$ and the key as $(1,0)$. Then, under the RoPE convention
used in this paper, the attention score for a token at relative distance
$u$ from the final position is
\[
\lambda_h(u)
=
\bigl[
[Ac_hc_u]_q+[As_hs_u]_q
\bigr]_q.
\]
The value matrix of the $h$-th head copies the one-hot representation of
the selected token into the $h$-th reserved block.

First consider an input $x=(x[1],\ldots,x[n])$ with $n\le L$. Since
$L/P\le1/16$, there is no phase wrap for $0\le h,u<L$. Let
$\eta=10^{-q}$. By the $q$-precision rounding of the trigonometric values
and the two multiplications,
\[
\left|
\lambda_h(u)
-
A\cos\left(\frac{u-h}{P}\right)
\right|
\le4A\eta.
\]

For the $h$-th head, the ideal score
$A\cos((u-h)/P)$ is maximized at $u_*=\min\{h,n-1\}$. We show that this
maximum is sufficiently separated from every other visible position.

If $u_*=h$, then $|(u-h)/P|\ge1/P$ for every $u\neq h$. Since
$1-\cos t\ge t^2/4$ for $|t|\le1$, the ideal score gap is at least
$A/(4P^2)$.

If $u_*=n-1<h$, the maximum is attained at the boundary of the visible
positions. Writing $a=(h-(n-1))/P$ and using that all relevant phases lie
in $[0,1]$,
\[
\cos a-\cos(a+P^{-1})
=
\int_a^{a+P^{-1}}\sin t\,dt
\ge
\frac{1}{4P^2}.
\]
Hence the ideal score gap is again at least $A/(4P^2)$.

Taking the rounding errors into account, the computed score gap between the
maximal position and any competing position is at least
\[
\frac{A}{4P^2}-8A10^{-q}.
\]
By the definition $A=64(q+1)P^2$ and assumption
\eqref{eq:copy-precision},
\[
\begin{aligned}
\frac{A}{4P^2}-8A10^{-q}
&\ge
16(q+1)-\frac{512}{10^4}\\
&\ge
15(q+1)
>
q\ln 10+\ln2.
\end{aligned}
\]
Therefore, after subtracting the maximal attention score, every nonmaximal
exponential satisfies $[e^{-\Delta}]_q=0$, where $\Delta$ is the
corresponding score gap.

The exponential corresponding to the maximal score is exactly one, so the
softmax denominator is also one. Thus the selected position receives
attention weight exactly one and all other positions receive weight zero.
Hence the $h$-th head exactly copies the token at relative position
$u_*=\min\{h,n-1\}$. Since the corresponding destination block is zero in
the embedding layer, the residual connection does not affect it. Therefore
\[
V_h=e_{x[n-\min\{h,n-1\}]}.
\]
In particular, $V_h=e_{x[n-h]}$ for every $0\le h\le n-1$, proving
\eqref{eq:copy-exact}.

It remains to consider inputs whose lengths may exceed $L$. Exact copying is
no longer required, but the attention computation must remain well defined.

For an arbitrary relative position $u$, instead of directly computing the
possibly unbounded product $u\omega$, set $r_0=0$ and recursively define
\[
r_{u+1}
=
\begin{cases}
r_u+\omega,
& r_u+\omega<2[\pi]_q,\\
r_u+\omega-2[\pi]_q,
& r_u+\omega\ge2[\pi]_q.
\end{cases}
\]
Since both $\omega$ and $2[\pi]_q$ are $q$-precision numbers, every $r_u$
remains a $q$-precision number and satisfies
$0\le r_u<2[\pi]_q<7$. This gives the same reduced phase as in
Definition~13 without forming an unbounded intermediate value. Negative
relative positions are handled analogously by modular negation. Hence all
trigonometric values and attention scores remain in the permitted numerical
range.

Finally, consider an arbitrary attention row. Let $E_i$ denote its rounded
exponentials and set $Z=[\sum_iE_i]_q$. At least one $E_i$ equals one. If
$Z=\mathrm{Inf}$, then by the finite-precision division convention all
attention weights are zero, so the corresponding output block is zero.

Otherwise, $Z$ is finite. Let $p_i=E_i/Z$, so that $\sum_i p_i=1$. If
$[p_i]_q>0$, then $p_i\ge\eta/2$, and therefore
\[
[p_i]_q
\le
p_i+\frac{\eta}{2}
\le
2p_i.
\]
Consequently, $\sum_i[p_i]_q\le2$. Since the value vectors are one-hot
vectors, all coordinates of the head output are nonnegative and its
$\ell_1$ norm is at most $2$. Thus \eqref{eq:copy-safe} holds for inputs
of arbitrary length.

The construction uses $L$ attention heads and $L$ disjoint blocks of size
$K$, so its width is $O(LK)$. Moreover, neither the precision condition nor
the bound in \eqref{eq:copy-safe} depends on $K$.
\end{proof}

We now prove Theorem~\ref{agent}.

\begin{proof}
Let $\Sigma=\Sigma_2\cup\{\sigma_b,\sigma_e,0,1,-1,\sigma_T\}$. Without
loss of generality, let $\sigma_{|\Sigma_2|}=B$ and $T=|\Sigma|$.

We construct the agent module by module.

\medskip
\noindent\textbf{Step 1: the decision Transformer $\mathcal{F}_d$.}

The role of $\mathcal{F}_d$ is only to decide whether the agent should
continue the simulation or terminate. Its required behavior is
$\sigma_T\oplus\sigma_b\oplus\sigma_e$ when $-1$ does not occur in the
input, and
$\sigma_T\oplus\sigma_b\oplus\sigma_e\oplus B\oplus(-1)$ when $-1$
does occur.

Let $\mathcal{F}_1$ be the Transformer from Lemma~\ref{lm1} that detects
$-1$ and outputs $\mathcal{F}_1(x)=I_{|\Sigma_2|}$ when $-1$ occurs,
and $0$ otherwise. Let $\mathcal{F}_2$ be the Transformer from
Lemma~\ref{lm1} that detects $\sigma_T$ and outputs
$\mathcal{F}_2(x)=I_{|\Sigma_2|+5}$ when $\sigma_T$ occurs, and $0$
otherwise.

Finally, construct a Transformer $\mathcal{F}_3$ with embedding dimension
$5$. Use $v_i=(1,0,0,0,0)$ for $\sigma_i\in\Sigma_2$, and
\[
v_T=(0,1,0,0,0),\qquad
v_b=(0,0,1,0,0),\qquad
v_e=(0,0,0,1,0),\qquad
v_{-1}=(0,0,0,0,1).
\]
In its hidden layer, set $Q=K=V=0$ and $\operatorname{FNN}=0$. Let
$h=x_{\mathrm{last}}$ denote the final-token embedding. Choose the output
map so that
\[
\begin{aligned}
\mathcal{F}_3(x)[|\Sigma_2|+6]&=3h[1],&
\mathcal{F}_3(x)[|\Sigma_2|+1]&=10h[2],\\
\mathcal{F}_3(x)[|\Sigma_2|+2]&=10h[3],&
\mathcal{F}_3(x)[0]&=h[4]+10h[5],\\
\mathcal{F}_3(x)[|\Sigma_2|+5]&=-10h[4],
\end{aligned}
\]
with all remaining coordinates equal to zero.

Set
\[
\mathcal{F}_d=\frac54\mathcal{F}_1+4\mathcal{F}_2+\mathcal{F}_3.
\]
By direct inspection of the output logits, $\mathcal{F}_d$ has the desired
decision behavior.

Define the stopping rule by
$v_{\mathrm{end}}(x)=f_{\mathrm{Judge},-1}(x)$. Thus
$v_{\mathrm{end}}(x)=1$ exactly when the last symbol is $-1$.

\medskip
\noindent\textbf{Step 2: the data-selection function $v_d$.}

Let $v_d=f_{\mathrm{Search}}$. For an encoded memory state,
\[
v_d\bigl(\mathcal{T}'(x)[i],\mathcal{F}_d(\mathcal{T}'(x)[i])\bigr)
=
v_d\bigl(\mathcal{T}'(x)[i],
\sigma_T\oplus\sigma_b\oplus\sigma_e\bigr)
=
\ell_i(x),
\]
where
\[
\ell_i(x)
:=
a_i\oplus\sigma_b\oplus B_{|Q|}(k)\oplus\sigma_e\oplus b_i,
\]
with $a_i=\mathcal{T}(x)[i][j-1]$ and
$b_i=\mathcal{T}(x)[i][j]$. Thus $\ell_i(x)$ is the local segment
surrounding the current head position, where $b_i$ is the symbol currently
scanned by the head and $a_i$ is its left neighbor.

For the initial raw memory $M=x$, the symbols $\sigma_b$ and $\sigma_e$
have not yet been inserted. By the no-match branch of
$f_{\mathrm{Search}}$,
\[
v_d\bigl(x,\sigma_T\oplus\sigma_b\oplus\sigma_e\bigr)
=
B\oplus x[1].
\]
The execution Transformer interprets this two-symbol input as
$a=B$, $k=1$, and $b=x[1]$, corresponding to the initial state $q_1$ of
the Turing machine.

\medskip
\noindent\textbf{Step 3: the execution Transformer $\mathcal{F}_c$.}

We construct $\mathcal{F}_c$ so that it memorizes one transition of the
Turing machine. Let
$\delta(q_k,b_i)=(q_{k'},c_i,m_i)$ with $m_i\in\{-1,0,1\}$, where
$q_{k'}$ is the next state, $c_i$ is the symbol written at the current
head position, and $m_i$ is the head-movement direction. Regard
$\tau_i(x):=(k',c_i,m_i)$ as the transition information computed in
Part~5.

Part~6 converts this transition information, together with the local
symbol $a_i$, into the local memory-update block
\[
R_{m_i}
=
\begin{cases}
S_{k'}\oplus a_i\oplus c_i, & m_i=-1,\\
a_i\oplus S_{k'}\oplus c_i, & m_i=0,\\
a_i\oplus c_i\oplus S_{k'}, & m_i=1.
\end{cases}
\]
Define $u_i(x):=R_{m_i}(x)\oplus\sigma_{m_i}$. The required
autoregressive behavior of the execution Transformer is therefore
\[
\widehat{\mathcal F}_c(\ell_i(x))
=
\sigma_T\oplus R_{m_i}(x)\oplus\sigma_{m_i}.
\]

If $\sigma_T$ has not yet appeared, Lemma~\ref{lm1} allows
$\mathcal{F}_c$ to output $\sigma_T$ first. It therefore remains to
construct the one-step transition map after the marker has appeared.

Let $L=\lceil\log_2|Q|\rceil$ and $d=L+4$. By Lemma~10, the relevant
tokens in the bounded execution context can be recovered into fixed
coordinate blocks of the last hidden state. The remaining computation is
implemented by a constant-depth ReLU network, divided into the following
seven parts.

\medskip
\noindent\textbf{Part 0: indicator gadgets.}

For integer-valued inputs, define
\[
h_a(u)
=
\operatorname{ReLU}(u-a+1)
-2\operatorname{ReLU}(u-a)
+\operatorname{ReLU}(u-a-1).
\]
Then $h_a(u)=1$ if $u=a$ and $0$ otherwise. We also use
$\operatorname{clip}(u)=\operatorname{ReLU}(u)-\operatorname{ReLU}(u-1)$
and, for $v\in[0,1]$ and $g\in\{0,1\}$,
$\operatorname{Gate}(v,g)=\operatorname{ReLU}(v+g-1)$.
These gadgets implement the case distinctions used below.

\medskip
\noindent\textbf{Part 1: distinguish the initial configuration.}

Let $V_j$ denote the copied token at relative offset $j$ given by
Lemma~10. Define $E_t=\operatorname{clip}((V_t)_{\sigma_T})$ and
\[
U_j
=
\operatorname{clip}\left(
\sum_{t=0}^{d+1}
\operatorname{Gate}
\bigl(\operatorname{clip}(V_{t+j}),E_t\bigr)
\right).
\]
Thus $U_j$ is the $j$-th token preceding the last $\sigma_T$. Set
$\nu=(U_2)_{\sigma_e}$. On every valid execution input, $\nu=0$ exactly
in the initial case and $\nu=1$ once the state block has been explicitly
encoded.

\medskip
\noindent\textbf{Part 2: isolate the current local configuration.}

In the normal case, $b=U_1$ and $a=U_{L+4}$, while the $j$-th bit of the
current state code is obtained from $U_{L+3-j}$ for $1\le j\le L$.
When $\nu=0$, the execution input represents the implicit initial
configuration, so we instead use $a=B$ and $k=1$, while still taking
$b=U_1$.

\medskip
\noindent\textbf{Part 3: compute the generation phase.}

Retain the indicators $E_0,E_1,\ldots,E_{d+1}$. On every valid
autoregressive prefix, exactly one $E_t$ equals $1$. It indicates that
exactly $t$ tokens have already been generated after $\sigma_T$ and will
be used in Part~6 to select the next output token.

\medskip
\noindent\textbf{Part 4: decode the current state.}

Let $u_1,\ldots,u_L$ be the recovered state bits and define
$k_{\mathrm{raw}}=\sum_{j=1}^L2^{L-j}u_j$. Using $\nu$, select
\[
k_{\mathrm{current}}
=
\begin{cases}
1, & \nu=0,\\
k_{\mathrm{raw}}, & \nu=1.
\end{cases}
\]
This selection is implemented by the gating gadgets from Part~0.

\medskip
\noindent\textbf{Part 5: apply the transition function.}

For each $(k,s)\in Q\times\Sigma_2$, define
\[
T_{k,s}
=
\operatorname{ReLU}
\bigl(
h_k(k_{\mathrm{current}})+b_s-1
\bigr),
\]
where $b_s$ is the $s$-coordinate of the one-hot representation of $b$.
Exactly one such indicator is active. If
$\delta(q_k,s)=(q_{k'},c,m)$, the corresponding transition-table entries
directly produce the encoded next state $B_{|Q|}(k')$, the written symbol
$c$, and the direction $m$. Since there are
$O(|Q||\Sigma_2|)$ state-symbol pairs, this part has width
$O(|Q||\Sigma_2|)$ and constant depth.

\medskip
\noindent\textbf{Part 6: reconstruct and serialize the encoded configuration.}

Let $S_{k'}=\sigma_b\oplus B_{|Q|}(k')\oplus\sigma_e$. As in the
original construction, define
\[
R_m
=
\begin{cases}
S_{k'}\oplus a\oplus c, & m=-1,\\
a\oplus S_{k'}\oplus c, & m=0,\\
a\oplus c\oplus S_{k'}, & m=1.
\end{cases}
\]
The transition indicators select the appropriate $R_m$.

Finally, the phase indicators $E_t$ serialize this block: when $E_t=1$,
the network outputs the $(t+1)$-st token of $R_m$; after all tokens of
$R_m$ have been generated, it outputs the direction token $\tau_m$,
followed by $\sigma_0$. If $\sigma_T$ has not yet appeared, the marker
detector instead outputs $\sigma_T$ first. Therefore
$\widehat F_c(\ell)=\sigma_T\oplus R_m\oplus\tau_m$.

\medskip
\noindent\textbf{Step 4: memory replacement and final output.}

Define
\[
v_r(x,z)
:=
\begin{cases}
B\oplus
f_{\mathrm{Replace}}
\bigl(x,\sigma_T\oplus\sigma_b\oplus\sigma_e,z\bigr)
\oplus B,
& \text{if both $\sigma_b$ and $\sigma_e$ occur in $x$},\\[1mm]
B\oplus z[2:|z|-1]\oplus x[2:|x|]\oplus B,
& \text{otherwise}.
\end{cases}
\]
The second case occurs only at the first execution step, when the memory
still consists of the unencoded input. Since
$\widehat F_c(\ell_0(x))
=\sigma_T\oplus R_{m_0}(x)\oplus\tau_{m_0}$, we have
\[
v_r\bigl(x,\widehat F_c(\ell_0(x))\bigr)
=
B\oplus R_{m_0}(x)\oplus x[2:|x|]\oplus B,
\]
which is an encoded representation of the configuration after the first
transition.

For every subsequent step, both $\sigma_b$ and $\sigma_e$ occur in the
memory. By Step~2,
$f_{\mathrm{Search}}(T'(x)[i],
\sigma_T\oplus\sigma_b\oplus\sigma_e)$ locates the current local
configuration
\[
\ell_i(x)
=
a_i\oplus\sigma_b\oplus B_{|Q|}(k)
\oplus\sigma_e\oplus b_i.
\]
By the construction in Part~6,
$\widehat F_c(\ell_i(x))
=\sigma_T\oplus R_{m_i}(x)\oplus\tau_{m_i}$. Hence, by the definition
of $f_{\mathrm{Replace}}$, the local block $\ell_i(x)$ is replaced by
$\widehat F_c(\ell_i(x))[2:|\widehat F_c(\ell_i(x))|-1]
=R_{m_i}(x)$, while all symbols outside this block are left unchanged.

Therefore, after each execution cycle, the memory is a blank-padded
encoding of the next Turing-machine configuration, namely,
$B^{a_{i+1}}\oplus T'(x)[i+1]\oplus B^{b_{i+1}}$ for some
$a_{i+1},b_{i+1}\ge1$. These additional boundary blanks do not change the
represented tape configuration and ensure that the head can always move
one position to the left or to the right.

Finally, let

\[
v_{\mathrm{out}}(x)
=
f_{\mathrm{Delete}}
\bigl(x,\sigma_b\oplus\sigma_e\oplus(-1)\oplus B\bigr),
\]

consisting of $B$'s. 
After these auxiliary symbols are removed, the
remaining sequence is exactly the output of the memorized Turing machine.

It remains to verify the computational-equivalence statement. Each
execution cycle of the agent memorizes exactly one transition of $\TM$.
Hence, on every stoppable input $x$, the agent executes step~(b2) at most
$t_{\TM}(x)+1$ times.

Throughout the computation, the memory length is at most
$O(t_{\TM}(x)+\len(x)+\log_2|Q|)$. When $|Q|$ and $|\Sigma_2|$ are
treated as constants, the depth, width, number of attention heads,
numerical precision, and per-invocation output length of $\F_d$ and
$\F_c$ are all constants independent of $x$. Therefore, each invocation
of the fixed Transformers can be carried out in time polynomial in the
current memory length. Since $v_d$, $v_r$, $v_{\rm end}$, and $v_{\out}$
are computable in linear time, and there are at most $t_{\TM}(x)+1$
execution cycles, the total running time of $A_q$ is polynomial in
$t_{\TM}(x)+\len(x)$.

Thus $A_q$ memorizes $\TM$ with at most polynomial overhead. Conversely,
the entire computation of $A_q$ consists only of finite-precision
Transformer evaluations, linear-time tool calls, and updates of finite
strings, and hence can itself be memorized by a Turing machine with at
most polynomial overhead. Together with $A_q(x)=\TM(x)$ for every
stoppable input $x$, this shows that when $|Q|$ and $|\Sigma_2|$ are
constants, $A_q$ is computationally equivalent to $\TM$.
\end{proof}

\subsection{Agent with One Transformer}
\label{thcap}

In this section, we show that the two transformers $\F_c,\F_d$ in the agent can be the same one.  We can consolidate the functions of $\F_c$ and $\F_d$ into a single transformer, $\F_s$, and include a suitable prompt so that $\F_s$ can distinguish whether it is performing decision-making or execution.

Write $A^s=(\F_s,\mathbf{M},v_{\rm{end}},v_d,v_r,v_{\out})$ to indicate that the agent uses the same transformer $\F_s$ in both the decision module and the execution module. Then we have that:

\begin{proposition}
    \label{agent1}
    If $|\Sigma|\ge|\Sigma_2|+7$ holds, there exists linear time computable functions $v_d$, $v_{\out}$, $v_{\rm{end}}$, $v_r$ such that for any Turing machine $\TM=(Q,\Sigma_1,\Sigma_2,\delta)$, there is an agent $A^s_q$ such that $A^s_q$ can memorize the Turing machine $\TM$, and $A^s_q$ satisfies that:

(1) $\F_s$ has constant depth $O(1)$, width $O(|Q|\cdot|\Sigma_2|)$, number of heads $O(\ln|Q|)$, and numerical precision $q=O(\ln|Q|)$.

(2) The agent executes step (b2) at most $t_\TM(x)+1$ times, memory module requires at most $O(t_{\TM}+\len(x)+\len(|Q|))$ space.
\end{proposition}

The proof follows the same construction as Theorem~\ref{agent}, with an
additional prompt symbol used to distinguish the decision module from the execution module.

We first prove a simple combination lemma.

\begin{lemma}
\label{lmm1}
Let $q\ge 6$ be fixed, and let $\mathcal Z\subseteq\Sigma^*$ be a class
of nonempty contexts, closed under taking nonempty prefixes, in which
a distinguished token $\sigma_k$ occurs at most once.
Suppose that two $q$-precision Transformers $\mathcal F^1$ and
$\mathcal F^2$ are well-defined on every $z\in\mathcal Z$, and that
\[
\mathcal F^1(z),\mathcal F^2(z)\in[0,1]^{|\Sigma|+1}.
\]
Assume that their positional-encoding blocks can be embedded into
disjoint blocks of a common RoPE matrix.

Then there exists a $q$-precision Transformer
\[
\mathcal F
=
U(\mathcal F^1,\mathcal F^2,\sigma_k)
\]
such that, for every $z\in\mathcal Z$,
\[
\mathcal F(z)
=
\begin{cases}
\mathcal F^1(z), & \sigma_k\text{ does not occur in }z,\\
\mathcal F^2(z), & \sigma_k\text{ occurs in }z.
\end{cases}
\]
Moreover, the construction introduces only constant-factor/additive
overhead in depth, width, and number of heads, and does not change the
numerical precision $q$.
\end{lemma}

\begin{proof}
By Lemma~G.1, there exists a $q$-precision Transformer
$D_{\sigma_k}$ such that
\[
D_{\sigma_k}(z)
=
\begin{cases}
0, & \sigma_k\text{ does not occur in }z,\\
1, & \sigma_k\text{ occurs in }z.
\end{cases}
\]
Write $D(z)=D_{\sigma_k}(z)$.

Run $\mathcal F^1$, $\mathcal F^2$, and $D$ in disjoint hidden-state
blocks, and pad the shallower networks by identity layers if necessary.
Define, coordinatewise,
\[
\mathcal F(z)
=
\operatorname{ReLU}
\bigl(\mathcal F^1(z)-D(z)\mathbf 1\bigr)
+
\operatorname{ReLU}
\bigl(\mathcal F^2(z)-(1-D(z))\mathbf 1\bigr).
\]

If $\sigma_k$ is absent, then $D(z)=0$, and since
$\mathcal F^1(z),\mathcal F^2(z)\in[0,1]^{|\Sigma|+1}$,
\[
\mathcal F(z)
=
\mathcal F^1(z).
\]
If $\sigma_k$ is present, then $D(z)=1$, and similarly
\[
\mathcal F(z)
=
\mathcal F^2(z).
\]

All additional operations involve only values in $[-1,1]$ and are exact
under the same $q$-precision arithmetic. Hence no increase of numerical
precision is required. The resource bounds follow from the parallel
composition of the three constant-overhead subnetworks.
\end{proof}

\begin{proof}
Let
\[
\Sigma
=
\Sigma_2
\cup
\{\sigma_p,\sigma_b,\sigma_e,0,1,-1,\sigma_T\}.
\]
Use the same constructions of the decision Transformer
$\mathcal F_d$, execution Transformer $\mathcal F_c$,
$v_{\rm end}$, $v_r$, and $v_{\rm out}$ as in the proof of
Theorem~\ref{agent}, with the same common precision
$q=O(\ln |Q|)$ chosen from the outset.

The only change is that the local input supplied to the execution module
is prefixed by $\sigma_p$. Define
\[
v_d(x,z)
=
\sigma_p
\oplus
f_{\rm Search}(x,z).
\]
Enlarge the copying window in the construction of $\mathcal F_c$ by one
position. Since the execution construction reads the local configuration
and the generated prefix only through fixed relative offsets, prefixing
$\sigma_p$ leaves all relevant relative offsets unchanged. Hence, for
every valid local input $\ell$ and every autoregressive prefix $y[1:j]$,
\[
\mathcal F_c
\bigl(
\sigma_p\oplus\ell\oplus y[1:j]
\bigr)
=
\mathcal F_c
\bigl(
\ell\oplus y[1:j]
\bigr).
\]

Let $\mathcal Z$ contain all decision contexts, prompted execution
contexts, and their nonempty autoregressive prefixes. By the constructions
in Theorem~\ref{agent}, both branches are well-defined on $\mathcal Z$;
their candidate logits may be taken in $[0,1]^{|\Sigma|+1}$.
Therefore Lemma~\ref{lmm1} applies. Set
\[
\mathcal F_s
=
U(\mathcal F_d,\mathcal F_c,\sigma_p).
\]

On every decision call, $\sigma_p$ is absent, so $\mathcal F_s$ has the
same autoregressive output as $\mathcal F_d$. On every execution call,
$\sigma_p$ occurs exactly once, so $\mathcal F_s$ has the same
autoregressive output as $\mathcal F_c$.

The prompt symbol is not written into persistent memory. Hence
$v_r$, $v_{\rm end}$, and $v_{\rm out}$ are unchanged, and every cycle
produces the same memory update as in Theorem~\ref{agent}. Therefore, the
shared-transformer agent memorizes the same Turing-machine computation.

Finally, Lemma~\ref{lmm1} changes the depth, width, and head count only
by constant factors and does not change $q$. Thus, all resource bounds
of Theorem~\ref{agent} remain unchanged up to absolute constants.
\end{proof}

\section{Application Scopes of Results fir Finite Precision transformers, Infinite Precision transformers, and Agents}
\label{asct}

In this section, we will clarify the scope within which the theorems presented in this paper are valid.

{\bf Results in Section \ref{adfdfgsef} do not hold for finite precision transformers.}

Section \ref{adfdfgsef} considers the Turing incompleteness of infinite precision transformers on non-tame Turing machines, it holds for any $|\Sigma|$ such that $\Sigma_2\subset\Sigma_0$, but we must point out that these conclusions cannot be directly extended to finite precision transformers. Naturally, they cannot be applied to agents with finite transformers.

Results in Section \ref{sss1} do not hold for a finite precision transformer. 
Precision loss does not arise in an infinite-precision transformer and is therefore not taken into account in the sensitivity analysis of Theorem \ref{th-1}.
Once the loss of precision is taken into account, that bound will no longer hold.

Results in Section \ref{in-pr} are not valid for finite precision transformers. Under the precision limit, if the transformer can correctly output, the maximum weight is at least $10^{-q}$ larger than the others. By using several amplification layers, such a value can naturally reach any large amount not exceeding the precision limitation, so there is no confidence problem for a finite precision transformer.

Results in Section \ref{mtmfae} do not hold for a finite precision transformer. Under the precision limit, when the structure of the transformer and precision are fixed, the number of parameters in the hypothetical space is limited. Therefore, as long as there is at least one parameter that can make the transformer memorize the target Turing machine, the probability of memorizing such a Turing machine is not 0.

{\bf Comparing the results in Section \ref{ffa} and Section \ref{sec-agent}.}

Section \ref{ffa} considers the Turing incompleteness of finite precision transformers on non-converging Turing machines, it holds for any $|\Sigma|$ such that $\Sigma_2\subset\Sigma_0$. 
Since our conclusion in section \ref{ffa} holds for all finite precision transformers, we do not need to consider where the transformers come from, such as whether they are constructed or trained.

In Section \ref{sec-agent}, the transformer can complete the task through multiple cycles and calls to tools. 
In essence, it is a method to solve long tasks by dividing them into shorter problems, thereby avoiding having the transformer directly solve particularly difficult problems, such as non-converging Turing machines, which greatly reduces the difficulty of the tasks undertaken in the operation of the single-step transformer and ultimately realizes Turing completeness. 
So the results in Section \ref{ffa} do not contradict the Turing completeness of agents that use finite precision transformers.

In the results of agent, we require that $|\Sigma_2|+6\le|\Sigma|$, this is not a hard requirement: for $0-1$ Turing machine, $\Sigma_2=\{0,1,B\}$; for $\TM_{\Ar}$, $|\Sigma_2|=19$; on the other hand, $|\Sigma|$ is always a large value in most of transformer.


\section{Additional Experimental Details}

\subsection{Implementation Details}
\label{app:experimental-details}

\textbf{Datasets}
The benchmark contains 100 expressions for each of the 13 operator counts: $(1,2,3,4,5,6,8,10,12,15,20,25,30)$. 
Reference answers are generated by a deterministic evaluator and are represented as integers or reduced fractions. 
A prediction is correct only if its integer value, or both its numerator and denominator, exactly match the reference answer.
The following table collects several examples from our test datasets.

\begin{table}[h]
\centering
\tiny
\begin{tabular*}{\linewidth}{@{\hspace{5pt}\extracolsep{\fill}} l@{}l@{}c@{}c@{}c@{}l@{}l @{\hspace{5pt}}}
  \toprule
  id & input & answer & length & num\_operands & operator\_family & operators \\
  \midrule
  test\_short-mixed-l2-0017 & $14 - 10 / 1$ & 4 & 2 & 3 & mixed & ["-", "/"]\\
  test\_short-subtract-l3-0008 & $79 - 15 - 30 - 81$ & -47 & 3 & 4 & subtract & ["-", "-", "-"]\\
  test\_short-mixed-l3-0019 & $15 / 5 + 19 / 3$ & 28/3 & 3 & 4 & mixed & ["/", "+", "/"]\\
  test\_medium-mixed-l6-0016 & $4 + 9 * 19 - 14 * 6 * 4 - 6$ & -167 & 6 & 7 & mixed & ["+", "*", "-", "*", "*", "-"]\\
  \bottomrule
\end{tabular*}
\end{table}



\textbf{Models and Inference Configuration}
We use Llama3.1-8B, Qwen3-8B, and GLM4.5-air as our test models and evaluate instruction-tuned models served with vLLM through its OpenAI-compatible API.
For each evaluated model, the decision component $\F_d$ and execution module $\F_c$ in the agent use identical model weights, numerical precision, chat template, decoding parameters, and structured-output backend, which is the same as the single model.
The primary evaluation uses greedy decoding with temperature $0.7$, seed $42$, and pass@1. 
A single-model response is limited to 1,024 generated tokens, while each decision or execution response is limited to 512 generated tokens.  
To avoid the influence of random factors, we add a supplementary analysis using independent pass@\(k\) sampling, which is presented in Appendix~\ref{app:detailed-result-analysis}.

\textbf{Evaluation Metrics}
Our primary metric is exact-answer accuracy.  
Predictions are parsed as integers or reduced fractions before comparison, allowing mathematically equivalent fraction forms to match.  
We also report the completion rate, defined as the fraction of examples for which a system returns a parseable final answer.  
Both metrics are reported by checkpoint, operator count, length split, and operator family.
For the agent, we additionally report the mean number of model calls, invalid action rate, token usage, and failures caused by reaching a generation or round limit.  
Data and analysis are detailed in Appendix~\ref{app:detailed-result-analysis}.

\subsection{Example Agent Trajectory}
\label{app:agent-example}

The following table presents a successful trajectory sampled from the experiment using GLM-4.5-Air.  
The input expression is $15 / 5 + 19 / 3$, whose reference answer is $28/3$.  The table omits the models' textual reasoning and retains the actions that affect execution.

\begin{table}[h]
\centering
\small
\begin{tabular*}{\linewidth}{@{\hspace{5pt}\extracolsep{\fill}} c l l l l l @{\hspace{5pt}}}
  \toprule
  Round & Memory $\mathbf{M}_t$ & $\F_d$ & $v_d$ & $\F_c$ & Memory $\mathbf{M}_{t+1}$ \\
  \midrule
  1 &
  $15\,\underset{o_1}{\underline{\mathstrut /}}\,5\,\underset{o_2}{\underline{\mathstrut +}}\,19\,\underset{o_3}{\underline{\mathstrut /}}\,3$ &
  $o_1$ &
  $(15,/,5)$ &
  $3$ &
  $3\,\underset{o_2}{\underline{\mathstrut +}}\,19\,\underset{o_3}{\underline{\mathstrut /}}\,3$ \\
  \cmidrule(lr){1-6}
  2 &
  $3\,\underset{o_2}{\underline{\mathstrut +}}\,19\,\underset{o_3}{\underline{\mathstrut /}}\,3$ &
  $o_3$ &
  $(19,/,3)$ &
  $19/3$ &
  $3\,\underset{o_2}{\underline{\mathstrut +}}\,19/3$ \\
  \cmidrule(lr){1-6}
  3 &
  $3\,\underset{o_2}{\underline{\mathstrut +}}\,19/3$ &
  $o_2$ &
  $(3,+,19/3)$ &
  $28/3$ &
  $28/3$ \\
  \bottomrule
\end{tabular*}
\end{table}

In the first round, the decision invocation $F_d$ selects \(o_1\).
The read function \(v_d\) locates this stable identifier and returns only the local operation \(15/5\).  
An independent computation invocation \(F_c\) returns \(3\), which \(v_r\) writes back to memory.  
The same procedure next reduces \(19/3\) and finally adds the two intermediate values.  
The trajectory therefore uses three decision invocations and three computation invocations.

The identifiers remain attached to the surviving operators as the expression is reduced.  
The deterministic functions \(v_d\) and \(v_r\) only retrieve and replace memory entries: they do not select an operator, perform arithmetic, or
verify that a replacement is correct.  
In this trajectory, the final memory state contains \(28/3\), which matches the reference answer.

\subsection{Detailed Result Analysis}
\label{app:detailed-result-analysis}

Figure~\ref{fig:exp_results_app} reports Pass@8 accuracy and completion rate as functions of the number of operators for GLM4.5-air and Qwen3-8B. 
Pass@8 reduces the influence of single-sample stochasticity by evaluating eight sampled attempts per problem. 
In the accuracy panel, performance declines with expression length for all curves; at larger operator counts, however, the agent remains more accurate than the single-model baseline, which collapses toward zero. 
The completion-rate panel also exhibits length-dependent degradation, though its severity varies across models and settings.

This comparison reflects two opposing effects: iterative execution reduces the state-management burden of each model invocation, but each additional decision and execution call creates another opportunity for error. 
An incorrect operator choice or replacement is written to memory and affects all subsequent rounds. 
External memory therefore does not eliminate model error; its benefit depends on whether local decomposition outweighs the accumulation of per-step errors. 
Accuracy and completion rate diagnose different failure modes: low completion indicates protocol errors, malformed actions, or generation truncation, whereas high completion with low accuracy points to semantic errors in operator selection or arithmetic. 
Consequently, the agent may match or underperform the single model on short expressions without contradicting the theoretical result. 
The theory concerns the computational capability enabled by repeated invocation and external memory, rather than uniform empirical dominance on every finite input.




\usepgfplotslibrary{groupplots}
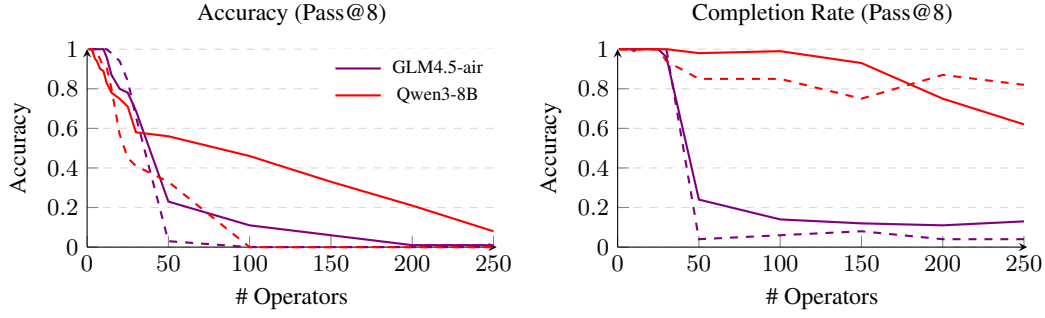
\begin{figure*}[t]
    \footnotesize
    \centering
    \begin{tikzpicture}
        \begin{groupplot}[
            group style={
                group size=2 by 1,
                horizontal sep=1.6cm,
                vertical sep=2.2cm,
            },
            width=7cm,
            height=4.2cm,
            axis lines=left,
            ylabel={Accuracy},
            tick align=inside,
            ymajorgrids=true,
            grid style={dashed, gray!30},
        ]

        \nextgroupplot[
            xlabel={\# Operators}, 
            title={\footnotesize Accuracy (Pass@8)}, 
            xmin=0, xmax=252, 
            ymin=0, ymax=1, 
            clip=false,
            legend style={
                draw=none,
                at={(1,1)},
                font=\scriptsize, 
                legend columns=1,
                column sep=2pt,
            }
        ]
        \addplot[violet, dashed, thick, no marks, forget plot] coordinates {
            (1,1)
            (2,1)
            (3,1)
            (4,1)
            (5,1)
            (6,1)
            (8,1)
            (10,1)
            (12,1)
            (15,0.98)
            (20,0.94)
            (25,0.84)
            (30,0.66)
            (50,0.03)
            (100,0)
            (150,0)
            (200,0)
            (250,0)
        };
        \addplot[violet, thick, no marks] coordinates {
            (1,1)
            (2,1)
            (3,1)
            (4,1)
            (5,1)
            (6,1)
            (8,1)
            (10,1)
            (12,0.96)
            (15,0.87)
            (20,0.8)
            (25,0.78)
            (30,0.69)
            (50,0.23)
            (100,0.11)
            (150,0.06)
            (200,0.01)
            (250,0.01)
        };
        \addplot[red, dashed, thick, no marks, forget plot] coordinates {
            (1,1)
            (2,1)
            (3,1)
            (4,0.98)
            (5,0.99)
            (6,0.98)
            (8,0.95)
            (10,0.91)
            (12,0.91)
            (15,0.8)
            (20,0.57)
            (25,0.45)
            (30,0.41)
            (50,0.33)
            (100,0)
            (150,0)
            (200,0)
            (250,0)
        };
        \addplot[red, thick, no marks] coordinates {
            (1,1)
            (2,1)
            (3,1)
            (4,0.97)
            (5,0.95)
            (6,0.94)
            (8,0.9)
            (10,0.89)
            (12,0.83)
            (15,0.78)
            (20,0.75)
            (25,0.71)
            (30,0.58)
            (50,0.56)
            (100,0.46)
            (150,0.33)
            (200,0.21)
            (250,0.08)
        };
        \addlegendentry{GLM4.5-air};
        \addlegendentry{Qwen3-8B};

        \nextgroupplot[xlabel={\# Operators}, title={\footnotesize Completion Rate (Pass@8)}, xmin=0, xmax=252, ymin=0, ymax=1, clip=false]
        \addplot[violet, dashed, thick, no marks] coordinates {
            (1,1)
            (2,1)
            (3,1)
            (4,1)
            (5,1)
            (6,1)
            (8,1)
            (10,1)
            (12,1)
            (15,1)
            (20,1)
            (25,1)
            (30,1)
            (50,0.04)
            (100,0.06)
            (150,0.08)
            (200,0.04)
            (250,0.04)
        };
        \addplot[violet, thick, no marks] coordinates {
            (1,1)
            (2,1)
            (3,1)
            (4,1)
            (5,1)
            (6,1)
            (8,1)
            (10,1)
            (12,1)
            (15,1)
            (20,1)
            (25,1)
            (30,0.96)
            (50,0.24)
            (100,0.14)
            (150,0.12)
            (200,0.11)
            (250,0.13)
        };
        \addplot[red, dashed, thick, no marks] coordinates {
            (1,1)
            (2,1)
            (3,1)
            (4,1)
            (5,1)
            (6,1)
            (8,1)
            (10,1)
            (12,1)
            (15,1)
            (20,1)
            (25,1)
            (30,0.94)
            (50,0.85)
            (100,0.85)
            (150,0.75)
            (200,0.87)
            (250,0.82)
        };
        \addplot[red, thick, no marks] coordinates {
            (1,1)
            (2,1)
            (3,1)
            (4,1)
            (5,1)
            (6,1)
            (8,1)
            (10,0.99)
            (12,1)
            (15,1)
            (20,1)
            (25,1)
            (30,1)
            (50,0.98)
            (100,0.99)
            (150,0.93)
            (200,0.75)
            (250,0.62)
        };

        \end{groupplot}
    \end{tikzpicture}%
    \caption{Accuracy rate of single model and agent as the number of required operations grows.}
    \label{fig:exp_results_app}
\end{figure*}

\end{document}